\documentclass{amsart}
\usepackage{amssymb, latexsym}

\newcommand{\bm}[1]{\mb{\boldmath ${#1}$}}
\newcommand{\beqa}{\begin{eqnarray*}}
\newcommand{\eeqa}{\end{eqnarray*}}
\newcommand{\beqn}{\begin{eqnarray}}
\newcommand{\eeqn}{\end{eqnarray}}

\newcommand{\iy}{\infty}

\newcommand{\lt}{\left}
\newcommand{\rt}{\right}

\newcommand{\C}{\mathbb C}
\newcommand{\D}{\mathbb D}
\newcommand{\R}{\mathbb R}

\newcommand{\N}{\mathbb N}

\newcommand{\mcA}{\mathcal A}

\newcommand{\mcH}{\mathcal H}
\newcommand{\mcB}{\mathcal B}
\newcommand{\mcC}{\mathcal C}
\newcommand{\mcD}{\mathcal D}
\newcommand{\mcE}{\mathcal E}
\newcommand{\mcF}{\mathcal F}

\newcommand{\mcP}{\mathcal P}

\newcommand{\mcL}{\mathcal L}

\newcommand{\mcM}{\mathcal M}
\newcommand{\mfB}{\mathfrak B}
\newcommand{\mfC}{\mathfrak C}
\newcommand{\mfF}{\mathfrak F}

\newcommand{\mfL}{\mathfrak L}
\newcommand{\mfO}{\mathfrak O}

\newcommand{\mfX}{\mathfrak X}

\newcommand{\mfM}{\mathfrak M}

\newcommand{\mb}{\makebox}

\newcommand{\tf}{\tfrac}

\newcommand{\al}{\alpha}

\newcommand{\e}{\varepsilon}
\newcommand{\ph}{\phi}
\newcommand{\de}{\delta}

\newcommand{\la}{\lambda}

\newcommand{\ka}{\kappa}

\newcommand{\Om}{\Omega}

\newcommand{\s}{\sigma}

\newcommand{\lb}{\label}
\newcommand{\rf}{\ref}
\newcounter{cnt1}
\newcounter{cnt2}
\newcounter{cnt3}
\newcommand{\blr}{\begin{list}{$($\roman{cnt1}$)$}
 {\usecounter{cnt1} \setlength{\topsep}{0pt}
 \setlength{\itemsep}{0pt}}}
\newcommand{\bla}{\begin{list}{$($\alph{cnt2}$)$}
 {\usecounter{cnt2} \setlength{\topsep}{0pt}
 \setlength{\itemsep}{0pt}}}
\newcommand{\bln}{\begin{list}{$($\arabic{cnt3}$)$}
 {\usecounter{cnt3} \setlength{\topsep}{0pt}
 \setlength{\itemsep}{0pt}}}
\newcommand{\el}{\end{list}}
\newtheorem{thm}{Theorem}[section]
\newtheorem{lem}[thm]{Lemma}
\newtheorem{cor}[thm]{Corollary}

\newtheorem{Def}[thm]{Definition}

\newtheorem{rem}[thm]{Remark}
\newcommand{\Rem}{\begin{rem} \rm}
\newcommand{\bdfn}{\begin{Def} \rm}
\newcommand{\edfn}{\end{Def}}
\newcommand{\tx}{\text}

\newcommand{\ba}{\begin{array}}
\newcommand{\ea}{\end{array}}

\begin{document}

\title{On the physical and mathematical foundations of quantum
physics via functional integrals}
\author[Esposito]{Giampiero  Esposito}
\address[Giampiero Esposito]{ Dipartimento di Fisica ``Ettore Pancini'',
Complesso Universitario di Monte S. Angelo\\ Via Cintia Edificio 6, 80126 Napoli, Italy\\
INFN Sezione di Napoli, Complesso Universitario di Monte S. Angelo\\ 
Via Cintia Edificio 6, 80126 Napoli, Italy}
\author[Gill]{Tepper L. Gill}
\address[Tepper L. Gill]{Department of EECS, Mathematics and Computational Physics Laboratory \\ 
Howard University, Washington DC 20059, USA}

\begin{abstract}  
In order to preserve the leading role of the action principle in
formulating all field theories one needs quantum field theory, with
the associated BRST symmetry, and Feynman-DeWitt-Faddeev-Popov ghost
fields. Such fields result from the fibre-bundle structure of the space
of histories, but the physics-oriented literature used them formally
because a rigorous theory of measure and integration was lacking. 
Motivated by this framework, this paper exploits previous work of Gill
and Zachary, where the use of Banach spaces for the Feynman integral
was proposed. The Henstock-Kurzweil integral is first introduced, because it
makes it possible to integrate functions like $exp\{ix^2/2\}$.  
The Lebesgue measure on $\R^\iy$ is then built and used 
to define the measure on every separable Hilbert space.
The subsequent step is the construction of a 
new Hilbert space $KS^2[\R^n]$, which contains 
$L^2[\R^n]$ as a continuous dense embedding, and contains both the test functions 
$\mcD[\R^n]$ and their dual $\mcD^*[\R^n]$, the Schwartz space of distributions, 
as continuous embeddings. This space allows us to construct the Feynman path integral 
in a manner that maintains its intuitive and computational advantages. We also extend 
this space to $KS^2[\mcH]$, where $\mcH$ is any separable Hilbert space. 
Last, the existence of a unique universal definition of time, $\tau_h$, 
that we call historical time, is proven. We use $\tau_h$ as the order parameter 
for our construction of Feynman's time ordered operator calculus, which in turn 
is used to extend the path integral in order to include all time dependent groups and 
semigroups with a reproducing kernel representation.
\end{abstract}
\maketitle
\section{Introduction}

The present paper is devoted to the physical and mathematical foundations
of Feynman's functional integrals. On the physical side, one should
first acknowledge the early work of Dirac, who studied, first, the 
Lagrangian in quantum mechanics both in a paper \cite{DI1} and in his famous book
on the principles of quantum mechanics \cite{DI2}. Later on Feynman, who had studied
the action at a distance with Wheeler, became interested in
a model where two particles can interact with each other \cite{WH}, but cannot
interact with themselves \cite{FY1}. Feynman realized that he could not write down
Hamilton's equations of motion for these particles, and hence was led to
study a quantization of the two-particle system not relying upon
Hamiltonian methods, but rather using Lagrangian methods along the seminal
work by Dirac. Over more than two decades Feynman developed first what
he called a spacetime approach to nonrelativistic quantum mechanics \cite{FY2},
and then extended formally his approach to gauge field 
theories \cite{FY3,FY5}, discovering
eventually the fields \cite{FY4} that came to be known as 
Feynman-DeWitt-Faddeev-Popov ghost fields. 
After the pioneering work of Misner \cite{MI},
Feynman presented his idea at
a conference in Poland \cite{FY4}. Within a few years, Faddeev and Popov had developed
a formal derivation for quantum Yang-Mills theory \cite{FAD}, whereas DeWitt obtained
a systematic derivation for all gauge field theories \cite{DW1}, showing that ghost
fields arise from the fibre-bundle structure of the space of field histories,
when a spacetime \cite{DW3}, global approach to quantum field theory 
\cite{DW4} is developed.
With hindsight, the need for a quantum theory of gauge fields might
be described as follows.

When we impose a supplementary condition in field theory 
\cite{DW4} in
order to obtain an invertible operator on gauge fields, e.g.,
the Lorenz condition in electrodynamics
\begin{equation}
\Phi_{L}(A)=\sum_{\mu=1}^{4}\nabla^{\mu}A_{\mu}=\tau,
\label{(1)}
\end{equation}
or the de Donder condition for gravitation:
\begin{equation}
\Phi_{\mu}(h)=\sum_{\nu=1}^{4}\nabla^{\nu}
\left(h_{\mu \nu}-{\tf{1}{2} }
g_{\mu \nu}{\rm tr}(h)\right)=\tau_{\mu},
\label{(2)}
\end{equation}
where $h_{\mu \nu}$ are the components of perturbations of the
metric $g$, we are dealing with functionals that associate to
physical fields some equations which can be denoted by $P^{\alpha}$,
the $\alpha$ being Lie-algebra indices. At classical level one
can remark for example that the Maxwell action functional in
Minkowski spacetime with metric $\eta$ 
provides the noninvertible operator 
$$
Q_{\mu \nu}=-\eta_{\mu \nu}\Box 
+ \partial_{\mu} \partial_{\nu}.
$$
Such an operator acts on smooth sections of the bundle of $1$-form 
fields on Minkowski spacetime, as is clearer by writing it in the form 
$$
Q_{\mu}^{\; \nu}=-\delta_{\mu}^{\; \nu}\Box
+\partial_{\mu}\partial^{\nu}.
$$
We note now that an additional term in the action such as the square 
of the Lorenz gauge would provide a contribution
$-\partial_{\mu}\partial_{\nu}$ to such an operator,
by virtue of the identity
\begin{equation}
\Phi_{L}^{2}(A)=\sum_{\mu, \nu=1}^{4} \Bigr[\partial^{\mu}
(A_{\mu}\partial^{\nu}A_{\nu})-A_{\mu}\partial^{\mu}
\partial^{\nu}A_{\nu}\Bigr].
\label{(3)}
\end{equation}
The resulting operator on $A^{\nu}$ would then be 
$P_{\mu \nu}=-\eta_{\mu \nu}\Box$, which is instead
invertible. More generally, any operator of the form
(here $\rho \in {\mathbb R} - \{ 0 \}$)
\begin{equation}
P_{\mu \nu}=-\eta_{\mu \nu}\Box
+\left(1-{1 / \rho}\right)\partial_{\mu}
\partial_{\nu}
\label{(4)}
\end{equation}
is invertible because its symbol is the matrix
(here $k^{2}=\sum_{\mu,\nu=1}^{4} 
\eta_{\mu \nu}k^{\mu}k^{\nu}$)
\begin{equation}
\sigma_{\mu \nu}=\eta_{\mu \nu}k^{2}
-\left(1-{1 / \rho}\right)k_{\mu} k_{\nu},
\label{(5)}
\end{equation}
whose inverse is \cite{E} 
\begin{equation}
\Sigma^{\nu \lambda}=\frac{1}{k^{2}}
\eta^{\nu \lambda}
+(\rho-1)\frac{k^{\nu}k^{\lambda}}{k^{4}}.
\label{(6)}
\end{equation}
The expression of the additional term in the action for
arbitrary choice of supplementary condition is therefore
$$
{\tf{1}{2} }\sum_{\alpha,\beta} \int d^{4}x 
\int d^{4}x' \; P^{\alpha}(x) \omega_{\alpha \beta}(x,x')
P^{\beta}(x')={\tf{1}{2} } P^{\alpha}
\omega_{\alpha \beta'} P^{\beta'}.
$$
With this notation, repeated indices without summation 
symbol represent actually summation as well as 
integration \cite{DW3,DW4}.

Moreover, from the identity (Latin lower case indices being
used for fields)
\begin{equation}
\delta P^{\alpha}=P_{,i}^{\alpha}\delta \varphi^{i}
=P_{,i}^{\alpha}Q_{\beta}^{i}\delta \xi^{\beta}
={\widehat F}_{\; \beta}^{\alpha} \delta \xi^{\beta}
\label{(7)}
\end{equation}
we discover the existence of the operator defined as follows:
\begin{equation} 
{\widehat F}_{\; \beta}^{\alpha}[\varphi] 
\equiv P_{\; ,i}^{\alpha}[\varphi]
Q_{\; \beta}^{i}[\varphi].
\label{(8)}
\end{equation}
Such an operator should act upon fields $\chi_{\alpha}$ and
$\psi^{\beta}$, in general independent of each other. 
For example, in the case of electrodynamics, ${\widehat F}$ is
a map from the space of smooth functions on spacetime into itself:
$$
{\widehat F}: C^{\infty}(M,g) \rightarrow C^{\infty}(M,g).
$$
For gravitation, ${\widehat F}$ maps smooth vector fields $\sigma_1$
over spacetime into smooth vector fields $\sigma_{2}$ over spacetime
($T(M,g)$ being our notation for the tangent bundle of the spacetime
manifold $(M,g)$):
$$
\sigma_{1}: (M,g) \longrightarrow T(M,g), \;
\sigma_{2}: (M,g) \longrightarrow T(M,g),
\; {\widehat F}: \sigma_{1} \longrightarrow \sigma_{2}.
$$
In light of Eq. \eqref{(8)}
we are led to assume that one can build the action functional
\begin{equation}
{\widetilde S}[\varphi,\chi,\psi]
=S[\varphi]+{\tf{1}{2} }P^{\alpha}[\varphi]
\omega_{\alpha \beta'}P^{\beta'}[\varphi]
+\chi_{\alpha}{\widehat F}_{\; \beta'}^{\alpha}[\varphi]
\psi^{\beta'}.
\label{(9)}
\end{equation}
The term quadratic in the functionals $P^{\alpha}$ spoils
gauge invariance of the classical action, whereas the sum
of the three terms in Eq. \eqref{(9)} has a new invariance property
(see theorem below)
under a class of transformations which are called \cite{BE} 
Becchi-Rouet-Stora-Tyutin (hereafter BRST) transformation.
They can be written in the form \cite{DW2}
\begin{equation}
\delta \varphi^{i}=Q_{\; \alpha}^{i}[\varphi]\psi^{\alpha}
\delta \lambda,
\label{(10)}
\end{equation}
\begin{equation}
\delta \chi_{\alpha}=\omega_{\alpha \beta}P^{\beta}[\varphi]
\delta \lambda,
\label{(11)}
\end{equation}
\begin{equation}
\delta \psi^{\alpha}=-{\tf{1}{2} }C_{\; \beta \gamma}^{\alpha}
\psi^{\beta}\psi^{\gamma}\delta \lambda,
\label{(12)}
\end{equation}
where $\delta \lambda$ is a constant which commutes
with $\varphi^{i}$ and anticommutes with
$\chi_{\alpha}$ and $\psi^{\alpha}$. The 
$Q_{\; \alpha}^{i}[\varphi]$ are the generators of 
infinitesimal gauge transformations 
(we write $\delta_{G}$ in order to avoid confusion
with $\delta$ used for BRST in Eqs. 
\eqref{(10)}-\eqref{(12)}, and we write
$\delta_{G}\xi^{\alpha}$ to denote a set of linearly independent
group parameters): 
\begin{equation}
\delta_{G}\varphi^{i}=Q_{\; \alpha}^{i}[\varphi]
\delta_{G}\xi^{\alpha}.
\label{(13)}
\end{equation}
Such $Q_{\; \alpha}^{i}$ are restricted from the identity 
of group-theoretical nature \cite{DW4}
\begin{equation}
Q_{\; \alpha,j}^{i}[\varphi]Q_{\; \beta}^{j}[\varphi]
-Q_{\; \beta,j}^{i}[\varphi] Q_{\; \alpha}^{j}[\varphi]
=Q_{\; \gamma}^{i}[\varphi]C_{\; \alpha \beta}^{\gamma},
\label{(14)}
\end{equation}
where $C_{\; \beta \gamma}^{\alpha}$ are the structure constants
of the Lie algebra of the gauge group.
For local theories, the $Q_{\; \alpha}^{i}$ are linear combinations
of Dirac's $\delta$ and its derivatives.
The ghost operator is discovered from Eq. \eqref{(7)}, which shows how
the functionals $P^{\alpha}$ used to fix the supplementary conditions
are varying under gauge transformations.  
The classical roots of such a statement can be understood by
remarking for example that, under a gauge transformation of the
potential, the $1$-form $A$, the functional $\Phi_{L}$ of the Lorenz
gauge changes by the amount (on denoting by
$\varepsilon$ a function of class $C^{2}$ in the gauge
transformation of the components of $A$) 
\begin{equation}
\Phi_{L}(A)-\Phi_{L}(A+\nabla \varepsilon)=-\Box \varepsilon.
\label{(15)}
\end{equation}
Similarly, the functional \eqref{(2)} for the de Donder gauge varies
under infinitesimal diffeomorphisms according to the relation
\begin{equation}
\Phi_{\mu}(h)-\Phi_{\mu}(h+L_{X}g)=
-{\tf{1}{2} }
\sum_{\nu=1}^{4}\left(g_{\mu \nu}\Box
+R_{\mu \nu}\right)X^{\nu}.
\label{(16)}
\end{equation}
The second-order differential operators in such equations
are two different realizations of one and the same concept, 
i.e., the ghost operator \cite{DW3,DW4} defined in Eq. \eqref{(8)}.
Now we can state and prove a well-known key theorem, as follows.
\vskip 0.3cm
\noindent
{\bf Theorem on the existence of a BRST-invariant action functional}. 
The action functional defined in Eq. \eqref{(9)}
is invariant under the BRST transformations defined in Eqs.
\eqref{(10)}, \eqref{(11)} and \eqref{(12)}.
\vskip 0.3cm
\noindent
{\bf Proof}. First of all let us point out that, 
by virtue of Eqs. \eqref{(10)}-\eqref{(12)}, 
the infinitesimal BRST variation of 
${\widetilde S}[\varphi,\chi,\psi]$ takes the form
\begin{eqnarray}
\delta {\widetilde S}[\varphi,\chi,\psi]&
=& S_{,i}[\varphi]Q_{\; \alpha}^{i}[\varphi]
\psi^{\alpha}\delta \lambda +{\tf{1}{2}}P_{\; ,i}^{\alpha}[\varphi]
Q_{\; \gamma}^{i}[\varphi]\omega_{\alpha \beta}
P^{\beta}[\varphi]\psi^{\gamma}\delta \lambda 
\nonumber \\
&+& { \tf{1}{2} }P^{\beta}[\varphi]\omega_{\beta \alpha}
P_{\; ,i}^{\alpha}[\varphi]Q_{\; \gamma}^{i}[\varphi]
\psi^{\gamma} \delta \lambda 
+ \omega_{\alpha \beta}P^{\beta}[\varphi]\delta \lambda
P_{\; ,i}^{\alpha}[\varphi]Q_{\; \gamma}^{i}[\varphi]\psi^{\gamma} 
\nonumber \\
&+& \chi_{\alpha}P_{\; ,ij}^{\alpha}[\varphi]Q_{\; \gamma}^{j}[\varphi]
\psi^{\gamma} \delta \lambda Q_{\; \beta}^{i}[\varphi]\psi^{\beta}
\nonumber \\ 
&+& \chi_{\alpha}P_{\; ,i}^{\alpha}[\varphi]Q_{\; \beta,j}^{i}[\varphi]
Q_{\; \gamma}^{j}[\varphi]\psi^{\gamma}
\delta \lambda \psi^{\beta} 
\nonumber \\
&-& {\tf{1}{2} }\chi_{\alpha}P_{\; ,i}^{\alpha}[\varphi]
Q_{\; \beta}^{i}[\varphi]C_{\; \gamma \delta}^{\beta}
\psi^{\gamma}\psi^{\delta}\delta \lambda.
\label{(17)}
\end{eqnarray}
From the gauge invariance of the classical action one finds that
$$
S_{,i}[\varphi]Q_{\; \gamma}^{i}[\varphi]=0,
$$
because
$$
\delta_{G}S=S_{,i}[\varphi]\delta_{G}\varphi^{i}
=S_{,i}[\varphi]Q_{\; \alpha}^{i}[\varphi] \delta_{G}\xi^{\alpha}=0.
$$
Moreover, the sum of second, third and fourth term on the
right-hand side of Eq. \eqref{(17)} vanishes as well, because
$
(\delta \lambda)\psi^{\alpha}=-\psi^{\alpha}(\delta \lambda)
$,
and exploiting the symmetric nature of
$\omega_{\alpha \beta}$. The fifth term on the right-hand side
of Eq. \eqref{(17)} reduces to
$$
-\chi_{\alpha}P_{\; ,ij}^{\alpha}[\varphi]
Q_{\; (\gamma}^{j}[\varphi]
Q_{\; \beta)}^{i}[\varphi]\psi^{[\gamma} \; \psi^{\beta]}\delta \lambda
$$
and hence vanishes as well. Last, upon using the identity 
\eqref{(14)} in order to express 
$Q_{\; \beta}^{i}[\varphi]C_{\; \gamma \delta}^{\beta}$, the sum 
of sixth and seventh term on the right-hand side of 
\eqref{(17)} is given by
$$ 
\chi_{\alpha}P_{\; ,i}^{\alpha}[\varphi]
Q_{\; \beta,j}^{i}[\varphi]Q_{\; \gamma}^{j}[\varphi]\psi^{\beta}
\psi^{\gamma} \delta \lambda 
-{\tf{1}{2} }\chi_{\alpha}P_{\; ,i}^{\alpha}[\varphi]
Q_{\; \gamma , j}^{i}[\varphi]Q_{\; \delta}^{j}[\varphi]
\psi^{\gamma}\psi^{\delta}\delta \lambda 
$$
$$
+{\tf{1}{2} }\chi_{\alpha}P_{\; ,i}^{\alpha}[\varphi]
Q_{\; \gamma,j}^{i}[\varphi]Q_{\; \beta}^{j}[\varphi]
\psi^{\beta}\psi^{\gamma}\delta \lambda 
$$
which is found to vanish after relabelling indices 
and exploiting the identity
$\psi^{\beta}\psi^{\gamma}=-\psi^{\gamma}\psi^{\beta}$. Q.E.D.

Thus, the BRST transformations involve also anticommuting fields,
and hence the BRST symmetry is not classical. Conceptually,
an important conclusion is found to emerge \cite{DW4,E}:
\vskip 0.3cm
\noindent
(1) Relativity suggests using the action principle as a 
foundation for all field theories.
\vskip 0.3cm
\noindent
(2) The gauge principle leads to a gauge-invariant action 
functional $S$.
\vskip 0.3cm
\noindent
(3) The resulting operator $S_{,ij}$ on gauge fields 
(it maps indeed gauge fields into gauge fields) is not
invertible, so that no Green function can be defined for
$S_{,ij}$. 
\vskip 0.3cm
\noindent
(4) Such a drawback is amended by adding to the action the
second term on the right-hand side of Eq. \eqref{(9)}. 
\vskip 0.3cm
\noindent
(5) The full action becomes the functional ${\widetilde S}$
of Eq. \eqref{(9)}, which is no longer gauge-invariant, but rather
BRST-invariant. This is a quantum symmetry, since it needs
anticommuting fields associated to bosonic fields.
\vskip 0.3cm
\noindent
(6) With hindsight, {\it quantum field theory turns
out to be the branch of physics 
which is needed in order to preserve the guiding role of the
action principle, and the driving force is provided by relativity,
since it puts the emphasis on constructing the action
upon completing the original gauge-invariant action by
the addition of the other terms in Eq. \eqref{(9)}}.

In quantum field theory, the anticommuting fields discussed
before are accommodated within a global approach that relies
upon functional integrals. It is therefore compelling to
build a mathematically consistent theory of such 
integrals \cite{GZ09}, which is the topic of our paper. 

\subsection{Outline of this paper}

After some preliminaries, in section 2, we introduce the Henstock-Kurzweil 
integral (HK). This integral extends those of Bochner and Pettis and is defined for 
operator-valued functions that are not separately valued (where both the Bochner and 
Pettis integrals are not defined). Intuitively, one uses a version of the Riemann 
integral (of calculus) with the interior points chosen first, while the size of the 
base rectangle around any interior point is determined by an arbitrary positive function 
defined at that point (see also Appendix). It is important to note that it integrates 
functions such as $exp\{ix^2/2\}$. This integral was discovered  by 
Henstock \cite{HS1}, \cite{HS2}, and Kurzweil \cite{KW}.  
\medskip

In Section 3 we construct the Lebesgue measure on $\R^\iy$ by using its equivalent 
space $\R_I^\iy$, which is defined in the section. This approach 
makes it possible for us to define Lebesgue 
measure on every separable Hilbert space in Section 4.
\smallskip

In Section 5 we construct a new Hilbert space $KS^2[\R_I^n]$, which contains 
$L^2[\R_I^n]$ as a continuous dense embedding, and contains both the test functions 
$\mcD[\R_I^n]$ and their dual $\mcD^*[\R_I^n]$, the Schwartz space of distributions, 
as continuous embeddings. This space allows us to construct the Feynman path integral 
in a manner that maintains its intuitive \cite{GMG} 
and computational advantages. We also extend 
this space to $KS^2[\mcH]$, where $\mcH$ is a separable Hilbert space. 
\smallskip

In Section 6 we prove that a unique universal definition of time, $\tau_h$, 
exists (which we call historical time). We use $\tau_h$ as the order parameter 
for the construction of Feynman's time ordered operator calculus, which in turn 
is used to extend the path integral to include all time dependent groups and 
semigroups with a kernel.

Section 7 is devoted to concluding remarks, and some basic concepts of
real and functional analysis are further described in the Appendix.
It has been our choice to describe modern analysis concepts which, despite
being widespead in mathematics, are often ignored in the physics-oriented
literature. We think that their potentialities deserve greater attention in 
theoretical physics.

\subsection{{\bf Preliminaries}}

\subsubsection{\bf{Extensions and Cubes}}

Before we discuss the Henstock-Kurzweil integral on $\R^n$, we need an extension 
theorem for functions defined on a domain of $\R^n$ and a result that shows that a 
domain in $\R^n$ can be written as a union of nonoverlapping closed cubes. 
(proofs of these results may be found in Evans \cite{EV} and Stein \cite{STE}, respectively.)
Let $\D$ be a bounded open connected set of $\R^n$ (a domain) with topological boundary 
$\partial{\D}$ and closure $\bar{\mathbb{D}}$. 
\begin{Def} Let $k$ be a positive integer: we say that $\partial \D$ is of class 
${\bf{C}}^k$ if, for every point $\bf x \in \partial{\D}$, there is a homeomorphism 
$\varphi$ of a neighborhood $U$ of $\bf x$ into $\R^n$ such that both $\varphi$ and 
$\varphi^{-1}$  have $k$ continuous derivatives with  
\[
\varphi \left( {\mathbb{D} \cap U} \right) \subset \left\{ {(x_1 , 
\ldots ,x_n ) \in \mathbb{R}^n :x_n  > 0} \right\}
\]
and
\[
\varphi \left( {\partial \mathbb{D} \cap U} \right) \subset \left\{ {(x_1 , 
\ldots ,x_n ) \in \mathbb{R}^n :x_n  = 0} \right\}.
\]
\end{Def}
\begin{thm}\lb{d1} Let $\D$ be a domain in $\R^n$ with $\partial \D$ of class 
${\bf{C}}^1$. Let $\mathbb{U}$ be any bounded open set such that the closure 
of $\D$ is a compact subset of $\mathbb{U}$. Then there exists a linear operator 
$\mathfrak{E}$ mapping functions on $\D$ to functions on $\R^n$ such that:
\begin{enumerate}  
\item $\mathfrak{C}$ maps ${{W}}^{1,2}({\mathbb D})$ (see below) continuously 
into ${{W}}^{1,2}({\R}^n)$.
\item  $\mathfrak{C}(f)  \left| {_\mathbb{D}}=f \right.$ (e. g., 
$\mathfrak{E}(\cdot)$ is an extension operator).
\item  $\mathfrak{E}(f)(x)=0$ for $x \in {\mathbb{U}}^c$  (e.g., 
$\mathfrak{E}(f)$ has support inside ${\mathbb{U}}$).
\end{enumerate}
\end{thm}
\begin{thm}\lb{d2} Let $\D$ be a domain in $\R^n$. Then $\D$ is the union of a 
sequence of closed cubes $\{{\D}_k\}$ whose sides are parallel to the coordinate 
axes and whose interiors are mutually disjoint. 
\end{thm} 
\begin{rem}
Thus, if a function $f$ is defined on a domain in $\mathbb{R}^n$, by Theorem 
\rf{d1} it can be extended to the whole space. By Theorem \rf{d2}, without loss of generality, 
 we can assume that the domain is a cube with sides parallel 
to the coordinate axes. In either case, the HK-integral can be constructed under these conditions. 
\end{rem}

\subsubsection{\bf{Distributions}}

Let $\mathcal{D}({\mathbb{R}}^{n})={{\C}}_c^{\infty}({\mathbb{R}}^{n})$ denote the 
space of infinitely differentiable functions $\phi : {\mathbb{R}}^{n} \rightarrow 
\mathbb{R}$ with compact support. We say that a sequence of functions $\{ \phi_n \} 
\subset \mathcal{D}$ converges to $\phi \in \mathcal{D}$ if there is a fixed compact 
set $U$ such that all functions $\phi_n$ have their support in $U$ and, for each 
$k \ge 0$, the sequence of $k$-derivatives of $\phi_n,\; \phi_n^{(k)}$, converges 
uniformly \cite{FMS} to $\phi^{(k)}$ on $U$. We call a function $\phi$ belonging to 
$\mathcal{D}({\mathbb{R}}^{n})$ a test function.  

Let $u \in {{\C}}^{1}({\mathbb{R}}^{n})$. Then, if $\phi \in 
{{\C}}_c^{\infty}({\mathbb{R}}^{n})$, integration by parts gives:
\[
\int_{\mathbb{R}^n } {(u\phi _{y_i }) } d\lambda  = \int_{\partial 
U } {\left( {u\phi } \right)} \nu _i d{\mathbf{S}}   
- \int_{\mathbb{R}^n } {(\phi u_{y_i } )} d\lambda, \; 1 \le i \le n,
\]
where $\nu$ is the unit outward normal to $U$. Since $\phi$ vanishes 
on the boundary, we see that the above reduces to:
\[
\int_{\mathbb{R}^n } {(u\phi _{y_i } )} d\lambda =   - \int_{\mathbb{R}^n } 
{(\phi u_{y_i }) } \ d\lambda, \; 1 \le i \le n.
\]
In the general case, for any $u \in {{\C}}^{m}({\mathbb{R}}^{n})$ and any 
multi-index $\al = (\al_1,\dots, \al_n),\; \left| \alpha  \right| 
= \sum _{\al  = 1}^n \al _i = m$, we have
\[
\int_{\mathbb{R}^n } {u(D^{\al}\phi)} d\lambda 
=   (-1)^m \int_{\mathbb{R}^n } {\phi (D^{\al}u)} d\lambda.
\]
\begin{Def} If $\al$ is a multi-index and $u,\ v \in {{L}}_{loc}^1({\mathbb{R}}^{n})$,  
we say that $v$ is the $\al^{th}$-weak (or distributional) partial derivative 
of $u$ and write $D^{\al}u=v$ provided that
\[
\int_{\mathbb{R}^n } {u(D^{\al}\phi)} d\lambda =   (-1)^{\left| \alpha  \right|} 
\int_{\mathbb{R}^n } {\phi v} \ d\lambda
\]
for all functions $\phi \in {{\C}}_c^{\infty}({\mathbb{R}}^{n})$.  Thus, 
$v$ is in the dual space ${\mathcal{D}}^*(\R^n)$ of ${\mathcal{D}}(\R^n)$. 
\end{Def}
The next result is easy. 
\begin{lem} If a  weak $\al^{th}$-partial derivative exists for $u$, then it is a unique 
$\la$-a.e. (i.e., except for a set of measure zero).
\end{lem}
\begin{Def} If $m \ge 0$ is fixed and $1 \le p \le \infty$, we define 
the Sobolev space ${{W}}^{m,2}({\R}^n)$ as the set of all locally summable 
functions $u:\ {\R}^n \rightarrow \R$ such that, for each multi-index $\al$ 
with $\left| \alpha  \right| \leqslant m, \ D^{\al}u$ exists in the weak 
sense and belongs to ${{L}}^{2}({\R}^n)$. At a deeper level, Sobolev spaces 
can be defined in at least two different ways \cite{ACM}. On the one hand, 
once an open domain $\Omega$ of $\R^{n}$ is given and $p \in [1,\infty)$ 
is fixed, we can start from the space $C^{1}(\bar{\Omega};\R)$ which is the 
subset of $C^{1}(\Omega;\R)$ consisting of functions $u$ such that both $u$
and its gradient admit a continuous extension to the closure of $\Omega$.
We then consider the subspace of regular functions $u \in C^{1}({\bar \Omega};\R)$ 
such that the norm
\begin{equation}
\left \| u \right \|_{W^{1,p}(\Omega;\R)}
=\left[\left( \left \| u \right \|_{L^{p}(\Omega;\R)}\right)^{p}
+\left( \left \| \nabla u \right \|_{L^{p}(\Omega;\R)}\right)^{p}\right]^{\frac{1}{p}}
\end{equation}
is finite. The Sobolev space $H^{1,p}(\Omega;\R)$ is, by definition, the
completion of such a subspace of $C^{1}(\bar{\Omega};\R)$ with respect to 
the $W^{1,p}$ norm. Alternatively, for $p \in [1,\infty]$, the Sobolev space
$W^{1,p}(\Omega;\R)$ is the subset of $L^{p}(\Omega;\R)$ whose elements are 
weakly differentiable with corresponding derivatives also belonging to
$L^{p}(\Omega;\R)$.
\end{Def}
\begin{Def} If $\D$ is a domain in $\R^n$, we define ${{W}}_0^{m,2}({\mathbb D})$ 
as the closure of ${{\C}}_c^{\infty}({\D})$ in ${{W}}^{m,2}({\mathbb D})$. 
\end{Def}
\begin{rem}Thus, ${{W}}_0^{m,2}({\mathbb D})$ contains the functions 
$u \in {{W}}^{m,2}({\mathbb D})$ such that, for all $\lt|\al\rt| \le m-1, 
\; D^{\al}u =0$ on the boundary of $\D, \; \partial{\D}$.

We further note that it is also standard to use $H^m(\D)= {{W}}^{m,2}({\mathbb D})$ 
and $H_0^m(\D) ={{W}}_0^{m,2}({\mathbb D})$.
\end{rem}

\section{{The Henstock-Kurzweil Integral}}

In this section we begin with an elementary introduction to integrals 
for functions defined on an interval based on Riemann, Darboux, 
Kurzweil and Henstock. We then discuss the Henstock-Kurzweil  
(HK) integral for functions defined on $\R^n$ and for operator-valued 
functions (see \cite{HS1} and \cite{RUD}).  

\subsection{Riemann's integral}

Riemann provided the first rigorous definition of an integral for a real-valued function in 1868. 
\begin{Def}{\tx{\bf(Riemann)}} Let $f$ be a bounded real-valued function 
defined in the interval $-\iy<a<b< \iy$. For each partition 
$P=\{a=x_0 <x_1 \dots <x_n=b \}$ and each choice of $t_j \in [x_{j-1}, x_j]$, 
define the corresponding Riemann sum by
\[
S(f,P,\{ {t_j}\} ) = \sum\nolimits_{j = 1}^n {f({t_j})\Delta {x_j}}, 
\]  
where ${\Delta {x_j}}=x_j-x_{j-1}$. Define the norm of $P$ by  by 
\[
\left\| P \right\| = \max \left\{ {\Delta {x_j},\;1 \leqslant j \leqslant n} \right\}.
\]
We say that $f$ is Riemann integrable over $[a, b]$ if there exists a number 
$I$ such that, for each $\e>0$, there exists a $\de>0$ such that, 
whenever $\left\| P \right\|< \de$, we have
\[
\left| {I - S(f,P,\{ {t_j}\} )} \right| < \varepsilon. 
\] 
In this case, we write:
\[
I = R\int_a^b {f(x)dx}. 
\]
\end{Def}
\begin{rem}
The Riemann integral is used for many applications and numerical approximation purposes.  
We can also allow $f$ to be complex-valued on $[a, b]$.  
\end{rem}

\subsection{Darboux integral}

When $f$ is a real-valued function on $[a,b]$, an integral equivalent to 
the Riemann integral was defined by Darboux, which is often discussed in elementary calculus.  
\begin{Def} Let 
\[
a=x_0 < x_1 < x_2 < \dots < x_{n-1}< x_n =b
\]
be a  partition $P$ of $[a,b]$ and, set  $\Delta {x_j} = {x_j} - {x_{j - 1}}$. Define
\[
{M_j} = \mathop {\sup }\limits_{{x_{j - 1}} \leqslant t \leqslant 
{x_j}} \left\{ {f(t)} \right\} \quad {\tx{and}}\quad {m_j} = \mathop 
{\inf }\limits_{{x_{j - 1}} \leqslant t \leqslant {x_j}} \left\{ {f(t)} \right\}
\]
and define the upper and lower Darboux sums by:
\[
{I^u} = \sum\nolimits_{j = 1}^n {{M_j}} \Delta {x_j}\quad {\text{and}}\quad 
{I_l} = \sum\nolimits_{j = 1}^n {{m_j}} \Delta {x_j}.
\]
If the following limits exist
\[
\mathop {\lim }\limits_{\left\| P \right\| \to 0} {I^u} = \mathop 
{\lim }\limits_{\left\| P \right\| \to 0} {I_l}.
\]
We say that the Darboux integral exists for the function $f(x)$ on $[a,b]$ and write:
\[
I = D\int_a^b {f(x)dx}=\mathop {\lim }\limits_{\left\| P \right\| \to 0} {I^u}.
\] 
\end{Def}
The most important result of elementary calculus is the Fundamental 
Theorem of Calculus, which relates differentiation to integration.
\begin{thm}{\rm{(Fundamental Theorem of Calculus, Riemann)}}Suppose the 
derivative $F'(x)=f(x)$ is Riemann integrable on $[a, b]$,  then  
$\int_a^x {f(t)dt} =F(x)-F(a)$, for all $x \in [a, b]$.
\end{thm} 
The work of Baire, Borel, Cantor and others showed that the concepts of 
length and area were, like that of number, more delicate than expected, 
and required a deeper investigation. These issues have led to a closer 
study of differentiation and integration. In the early 1900s, Lebesgue 
weakened the conditions for the Fundamental Theorem. This led to the 
Lebesgue integral and the following:  
\begin{thm}{\rm{(Fundamental Theorem of Calculus, Lebesgue)}}Suppose 
the derivative $F'(x)=f(x)$ exists on $[a, b]$ and $F'(x)$ is bounded,  
then the Lebesgue integral exists and $L\int_a^x {f(t)dt} =F(x)-F(a)$, for all $x \in [a, b]$.
\end{thm} 
One is naturally led to ask the following question: Is it
possible to define an integration process (?) for which the following holds.
\begin{thm}
If $F$ is differentiable in $[a,b]$, then $F'$ is
(?)-integrable in $[a,b]$ and the Fundamental Theorem is valid.
\end{thm}
Three integration processes exist that
accomplish this task, and they were developed by
Denjoy, Perron and later by Kurzweil and Henstock (HK) (discussed below). 
The HK approach is closest to the Riemann integration process, while 
retaining all the advantages of Lebesgue's theory. Teaching in basic courses 
still focuses (mainly) on the Riemann and Lebesgue theories. However, the 
concept Lebesgue measure is very useful in analysis and geometry.  In addition, it can be easily 
extended to more general settings, whereas the corresponding idea 
associated with the HK integral does not extend as well \cite{GO}. 

\subsection{The Henstock-Kurzweil Integral}

While studying a generalized approach to ordinary differential equations 
Kurzweil noticed that, if the base intervals were chosen to be small where 
the function was steep, and large where it was flat, this was sufficient 
to extend the Riemann process to obtain a stronger integral than that of 
Lebesgue (see \cite{KW1}). Independently, Henstock discovered the 
same approach and provided the first systematic study of the subject (see  \cite{HS1}).

\subsubsection{\bf{One-dimensional HK-integral}}  

Let $\de(t)>0$ be a function defined on the compact interval $[a,b]$ and let 
$t_1<t_2 \cdots< t_n$ be points in the open interval (a, b). The base 
intervals $[x_{j-1},x_{j}]$ used to approximate the integral are chosen 
such that, $\left[ {{x_{j - 1}},\;{x_j}} \right] \subseteq \left( {{t_j} 
- \delta ({t_j}),\;{t_j} + \delta ({t_j})} \right)$. We call each 
$(t_j, [x_{j-1}, x_j])$ a tagged interval subordinate to $\de$ and the collection:
\[
\mcP = \left\{ {\left( {{t_j},\;\left[ {{x_{j - 1}},\;{x_j}} \right]} \right),
\;1 \leqslant j \leqslant n} \right\}
\]
a HK-$\de$ partition of $[a, b]$, if $[a,b] =  \cup _{j = 1}^n
\left[ {{x_{j - 1}},\;{x_j}} \right]$.
\begin{Def} We say that a function $f(t)$ defined on $[a,b]$ is HK integrable with value $I$,
if for each $\e>0$, there exists a function $\de(t)$ and a  HK-$\de$ partition of $[a, b]$ such that:
\[\left| {I - \sum\nolimits_{j = 1}^n {f({t_j})\Delta {x_j}} } \right| < \varepsilon \]
and we write $I= HK\int_a^b {f(x)dx}$.
\end{Def}
The following theorem provides additional information.
\begin{thm} Let $f(t)\,:\, [a,b] \to \mathbb{R}$.
\begin{enumerate}
\item If $f(t)$ is Lebesgue integrable in $[a,b]$, then it is HK-integrable in $[a,b]$ and 
HK$\int_a^b {f(t)dt} =$L$\int_a^b {f(t)dt}$.
\item If $f(t)$ is HK-integrable in $[a,b]$, then $\left| {\int_a^b {f(x)dx} } \right| < \infty$. 
\item If $f(t)$ is HK-integrable and nonnegative in $[a,b]$, then 
it is Lebesgue integrable in $[a,b]$. 
\item If $f(t)$ is HK-integrable on every measurable subset of $[a,b]$, 
then it is Lebesgue integrable in $[a,b]$. 
\end{enumerate}
\end{thm}
\begin{rem}Item (4) for the HK-integral is equivalent to the statement that, 
if every rearrangement of a series converges, it is an absolutely convergent series.
\end{rem}
A function has a property (n.e.) if it has the property except for a countable number of points.
\begin{cor} Let $F\,:\,[a,b] \to \mathbb{R}$
be continuous. If $F$ is differentiable {\rm {(n.e.)}} on $[a,b]$, then 
$F'$ is HK-integrable and HK$\int_a^t {F'(s)ds}  = F(t)-F(a)$ for each $t \in [a,b]$.
\end{cor}
This last result shows the sense in which the HK-integral is the 
reverse of the derivative. This result is not true for Lebesgue integrals.  
The standard counterexample \cite{GO} 
is $F'(t)=2t sin(\pi/t^{2})-{(2\pi/t})cos(\pi/t^{2})$ 
for $t$ irrational in $(0,1)$ and equal to $0$ for $t$ rational in $(0,1)$.  
(It can be seen that $F(t)=t^2 sin(\pi/t^{2})$.)

\subsubsection{\bf{The $n$-dimensional HK-integral}} 

There are a number of ways to approach the 
HK-integral for $\R^n$. We follow the approach of 
Lee Tuo-Yeong  (see \cite{TY} and \cite{TY1}).  
All norms are equivalent on $\R^n$, however, for the  HK-integral the maximal 
norm $\left\| {\mathbf{x}} \right\| = \mathop {\max }\limits_{1 \leqslant 
k \leqslant n} \left| {x_k } \right|$, is natural. With this norm, the open ball 
$B'({\mathbf{x}},r)$, is a cube centered at ${\bf{x}}$ with sides parallel to the 
coordinate axis of length $2r$. (open interval when $n=1$.) If the open 
interval for side $i$ about $x_i$ is $(a_i, b_i)$, we can represent  
$B'({\mathbf{x}},r)=(J', {\bf x})$, where $J'=\prod _{i = 1}^n (a_i ,b_i)$. 
(For a closed ball $B({\mathbf{x}},r)$ about ${\mathbf{x}}$, we can represent 
it as $B({\mathbf{x}},r)=(J, {\bf x})$, with $J=\prod _{i = 1}^n [a_i ,b_i]$.)  
Let $\la_n[\,\cdot \,]$ be Lebesgue measure on $\R^n$.
\begin{Def} If $E$ is a compact ball in $\R^n$, a partition $\mcP$, of $E$ is 
a collection $\left\{ {\left( {J_i ,{\mathbf{x}}_i } \right):{\mathbf{x}}_i  
\in J_i ,\;1 \leqslant i \leqslant m} \right\}$, where $J_1, J_2 \dots J_m$ 
are non-overlapping intervals (i.e, $\lambda _n \left[ {J_i  \cap J_j} \right]  = 0,i \ne j$
) and  ${\bigcup\nolimits_{i = 1}^m {J_i } }= E$.   
\end{Def}
\begin{Def}If $\de$ is a positive function on $E$, we say that $\mcP$ is a  
HK-${\de}$ partition for $E$ if for each $i, \;  
J_i  \subset B' \left( {{\mathbf{x}}_i ,\delta ({\mathbf{x}}_i )} \right)$. 
\end{Def} 
\begin{rem}
The function $\de$ is also called a gauge\footnote{This has nothing to do
with the word gauge used for physical theories of fundamental interactions.} 
on $E$ and a HK-${\de}$ partition is known as a $\de$-fine partition of $E$.
\end{rem}
\begin{lem} {\rm Cousin's Lemma} If $\de( \cdot)$ is a positive function 
in $E$, then a HK-${\de}$ partition exists for $E$. 
\end{lem}
\begin{Def} A function $f:  E \to \R$ is said to be HK-integrable on $E$, if 
there exists a number $I$ such that for any $\e>0$ there is a 
HK-${\de}$ partition on $E$ such that 
\begin{equation}
\left| {\sum\limits_{i = 1}^m {f({\bf{x}}_i )\lambda _n [J_i ]}  - I} \right| < \varepsilon.
\label{(18)}
\end{equation}
In this case, we write 
\[
I = HK\int_E {f({\mathbf{x}})d\lambda _n ({\mathbf{x}})} .
\]
\end{Def}
The following theorem provides a constructive definition of what it means to say 
that a function is absolutely continuous: (A proof of the following can be found in \cite{GZ}.) 
\begin{thm} Let $f \in L^1[\R^n]$.  If $\e>0$, then there is a $\de>0$
such that, whenever $E$ is a measurable set with $\la_n[E] <\de$,
 \[
\left| {\int\limits_E {f({\bf{x}})d\lambda _n ({\bf{x}})} } \right| < \varepsilon .
\]
\end{thm}
The next result shows that the Lebesgue integral is a special case of the 
HK-integral, first proven by Davies and Schuss \cite{DZ}, but see \cite{GZ} .     
\begin{thm} If $E$ is a measurable subset of $\R^n$ and  
$f: E \to \R$, has a finite Lebesgue integral on $E$, then
\[
HK\int\limits_E {f({\bf{x}})d\lambda _n ({\bf{x}})}  
= L\int\limits_E {f({\bf{x}})d\lambda _n (x)}.
\]
\end{thm}

\subsection{Integration of Operator-valued Functions}

In this section, we extend the HK-integral to operator-valued functions\footnote{
We refer to a review by Wightman \cite{W} for the related concept of
quantum fields as operator-valued distributions}. 
Let $\Om$ be a subset of $\R$, $\mfB[\Om]$ be the set of Borel measurable subsets 
of $\Om$ and let $\la$ be the Lebesgue measure on $\R$. A function $A(t), t \in \Om$, 
defined on the measure space $(\Om, \mfB[\Om], \la)$ with values in $\mcB = L[\mcH]$, 
the bounded linear operators on a Hilbert space $\mcH$, are called operator-valued functions.  

The problem with integration for operator-valued functions is that the functions 
may not have a Lebesgue integral (see \cite{HP}, pg. 71-80). 
To understand this problem, we begin with:
\begin{Def}
The function $A(t)$ is said to be: 
\begin{enumerate}
\item almost surely separably valued (or essentially separably valued) if 
there exists a subset $N \subset \Om$ with $\la(N) = 0$ such that 
$A(\Om \setminus N) \subset \mcB$ is separable,
\item countably-valued if it assumes at most a countable number of values in 
$\mcB$, assuming each value $\ne 0$ on a measurable subset of $\Om$,  and
\item strongly measurable if there exists a sequence $\{ A_n(t)) \}$ of countably-valued 
functions converging (a.s.) to $A(t)$.
\item Bochner integrable if $\lt\|A(t)\rt\|_\mcB$ is Lebesgue integrable.
\item Gelfand-Pettis integrable if $\left\langle {A(t), \mcL } \right\rangle_\mcB$ is 
Lebesgue integrable for each bounded linear functional $\mcL \in \mcB^*$.
\end{enumerate}
\end{Def}
In order to define constructively the integral of $A(t)$, we must be able to 
approximate it with simple operator-valued functions in either a strong  
(Bochner) or weak sense (Gelfand-Pettis). However, the function must be 
countably-valued and strongly measurable in the first case or weakly measurable 
in the second case. In the case of current interest, $\mcB$ is not separable, and the family 
of operator-valued functions $A(t): \Om \to \mcB$ need not be almost separably valued, 
and hence need not be strongly or weakly measurable. In this section, we extend 
the HK integral to operator-valued functions on $\R$. Because this version will be used 
for the Feynman operator calculus, we provide some detail.  

Let $[a,b] \subset \mathbb{R}$ and, for each $t \in [a,b]$, let $A(t) \in 
L({\mathcal{H}})$ be a given family of operators.
 
Recall that, if $\delta (t)$ maps $[a,b] \to (0,\infty )$, and ${\mcP} 
= \{ t_0 ,\tau _1 ,t_1 ,\tau _2 , \cdots ,\tau _n ,t_n \},$ where $a = t_0  
\leqslant \tau _1  \leqslant t_1  \leqslant  \cdots  \leqslant \tau _n  
\leqslant t_n  = b$, we say it is a  HK-$\de$ partition provided that, for $0 
\leqslant i \leqslant n - 1, \; t_i ,t_{i + 1}  \in (\tau _{i + 1}  
- \delta (\tau _{i + 1} ),\tau _{i + 1}  + \delta (\tau _{i + 1} )).$
\begin{lem} Let $
\delta _1 (t)$ and $\delta _2 (t)$ map $[a,b] \to (0,\infty )$, and assume that $
\delta _1 (t) \leqslant \delta _2 (t).$ Thus, if  ${\mcP}_1$ is a HK-${{\de}_1}$ 
partition, it is also a HK-${{\de}_2}$ partition.
\end{lem}
\begin{Def} The family $A(t) $, $t \in [a,b] $, is said to have a (uniform) 
HK-integral if there is an operator $Q[a,b] $ in $L({\mathcal{H}})$ such that, 
for each $\varepsilon  > 0$, there exists a $\delta$ and a HK-$\de$ partition such that
\[
\left\| {\sum\nolimits_{i = 1}^n {\Delta t_i A(\tau _i ) - Q[a,b]} } \right\| < \varepsilon. 
\]
In this case, we write
\[
Q[a,b] = (HK)\,\int_a^b {A(t)dt}. 
\]
\end{Def}
\begin{thm}For $t \in [a,b] $, suppose that the operators $A_1 (t) $ and 
$A_2 (t) $ both have HK-integrals, then their sum has as well and 
\[
(HK)\int_a^b {[A_1 (t) + A_2 (t)]dt}  = (HK)\,\int_a^b {A_1 (t)dt}  
+ (HK)\,\int_a^b {A_2 (t)dt}. 
\]
\end{thm}
\begin{thm} Suppose $\{ A_k (t)\left| {} \right.\,k \in {\mathbb{N}}\}$ is a 
family of operator-valued functions in $L[{\mathcal{H}}]$, converging 
uniformly to $A(t) $ on $[a,b] $, and 
$A_k (t) $ has a HK-integral $Q_k [a,b] $ for each $k$; then, $A(t)$ has a HK-integral 
$Q[a,b] $ and $Q_k [a,b] \to Q[a,b] $ uniformly.
\end{thm}
\begin{thm} Suppose $A(t)$ is Bochner integrable on $[a,b] $, then $A(t) $ has a HK-integral 
$Q[a,b]$ and:
\begin{equation}
(B)\,\,\int_a^b {A(t)dt}  = (HK)\,\int_a^b {A(t)dt}. 
\label{(19)}
\end{equation}
\end{thm}
\begin{proof} First, let $E$ be a measurable subset of $[a,b] $ and assume 
that $A(t) = A\chi _E (t) $, where $\chi _E (t) $ is the characteristic function 
of $E$.  In this case, we show that 
$Q[a,b] = A\la(E) $, where $\la(E)$ is the Lebesgue measure of $E$.  
Let $\varepsilon  > 0$ be given and let $D$ be a compact subset of $E$.  Let 
$F \subset [a,b] $ be an open set containing $E$ such that $\la(F\backslash D)
 < {\varepsilon  \mathord{\left/
 {\vphantom {\varepsilon  {\left\| A \right\|}}} \right.
 \kern-\nulldelimiterspace} {\left\| A \right\|}}$; and define 
$\delta \,:[a,b] \to (0,\infty ) $ such that:
\[
\delta (t) = \left\{ {\begin{array}{*{20}c}
   {d(t,[a,b]\backslash F),\;t \in E}  \\
   {d(t,D),\;t \in [a,b]\backslash E},  \\
\end{array} } \right.
\]
where $d(x\,,\,y) = \left| {x - y} \right|$ denotes the distance function.  Let $
{\mcP} = \{ t_0 ,\tau _1 ,t_1 ,\tau _2 , \cdots ,\tau _n ,t_n \}$ be a 
HK-$\de$ partition. If $\tau _i  \in E$ for $1 \leqslant i \leqslant n$,  
then $(t_{i - 1} ,t_i ) \subset F$ so that
\begin{equation}
\left\| {\sum\nolimits_{i = 1}^n {\Delta t_i A(\tau _i ) - A\la(F)} } \right\| 
= \left\| A \right\|\left[ {\la(F) - \sum\nolimits_{\tau _i  \in E} 
{\Delta t_i } } \right].
\label{(20)}
\end{equation}
On the other hand, if  $\tau _i  \notin E$ then $(t_{i - 1} ,t_i ) \cap D 
= \emptyset$ (empty set), then it follows that:
\begin{equation}
\left\| {\sum\nolimits_{i = 1}^n {\Delta t_i A(\tau _i ) - A\la(D)} } \right\| 
= \left\| A \right\|\left[ {\sum\nolimits_{\tau _i  \notin E} {\Delta t_i }  - \la(D)} \right].
\label{(21)}
\end{equation}
Combining equations \eqref{(20)} and \eqref{(21)}, we have that
\begin{eqnarray}
\begin{gathered}
  \left\| {\sum\nolimits_{i = 1}^n {\Delta t_i A(\tau _i ) - A\la(E)} } \right\| 
= \left\| A \right\|\left[ {\sum\nolimits_{\tau _i  \in E} 
{\Delta t_i }  - \la(E)} \right] \hfill \\
  {\text{                      }} \leqslant \left\| A \right\|\left[ {\la(F) 
- \la(E)} \right] \leqslant \left\| A \right\|\left[ {\la(F) - \la(D)} \right] 
\leqslant \left\| A \right\| \la(F\backslash D) < \varepsilon . \hfill \\ 
\end{gathered}
\label{(22)} 
\end{eqnarray}
Suppose that $A(t) = \sum\nolimits_{k = 1}^\infty  {A_k \chi _{E_k } (t)} $.  
By definition, $A(t) $ is Bochner integrable if and only if  $\left\| {A(t)} \right\|$ 
is Lebesgue integrable with: 
\[
(B)\int_a^b {A(t)} dt = \sum\nolimits_{k = 1}^\infty  {A_k \la(E_k )}, 
\]
and (cf. Hille and Phillips \cite{HP})
\[
(L)\int_a^b {\left\| {A(t)} \right\|} dt = \sum\nolimits_{k = 1}^\infty  
{\left\| {A_k } \right\|\la(E_k )}. 
\]
As the partial sums converge uniformly, $Q[a,b] $ exists and 
\[
Q[a,b] \equiv (HK)\int_a^b {A(t)} dt = (B)\int_a^b {A(t)} dt.
\]
Let $A(t) $ be an arbitrary Bochner integrable operator-valued function in 
$ L({\mathcal{H}})$, uniformly measurable and defined on $[a,b] $.  By definition, 
there exists a sequence $\{ A_k (t)\} $ of countably valued 
operator-valued functions in $L({\mathcal{H}})
$ that converges to $A(t) $ in the uniform operator topology such that:
\[
\mathop {\lim }\limits_{k \to \infty } (L)\int_a^b {\left\| {A_k (t) - A(t)} \right\|dt}  = 0,
\]
and
\[
(B)\int_a^b {A(t)dt}  = \mathop {\lim }\limits_{k \to \infty } (B)\int_a^b {A_k (t)dt}. 
\]
Since the $A_k (t) $ are countably-valued, 
\[
(HK)\int_a^b {A_k (t)dt}  = (B)\int_a^b {A_k (t)dt}, 
\]
so
\[
(B)\int_a^b {A(t)dt}  = \mathop {\lim }\limits_{k \to \infty } (HK)\int_a^b {A_k (t)dt}. 
\]
We are done if we show that $Q[a,b] $ exists.  Since every L-integral is a HK integral,  
$f_k (t) = \left\| {A_k (t) - A(t)} \right\|$ has a HK integral.  
This means that $\mathop {\lim }\limits_{k \to \infty } (HK)\int_a^b {f_k (t)dt}  = 0$.
Let $\varepsilon  > 0$ and let $m$ be so large that
\[
\left\| {(B)\int_a^b {A(t)dt}  - (HK)\int_a^b {A_m (t)dt} } \right\| < {\varepsilon  \mathord{\left/
 {\vphantom {\varepsilon  4}} \right.
 \kern-\nulldelimiterspace} 4}
\]
and
\[
(HK)\int_a^b {f_k (t)dt}  < {\varepsilon  \mathord{\left/
 {\vphantom {\varepsilon  4}} \right.
 \kern-\nulldelimiterspace} 4}.
\]
Choose $\delta _1$ so that, if $\{ t_0 ,\tau _1 ,t_1 ,\tau _2 , \cdots ,\tau _n ,t_n \}$ 
is a HK-${{\de_1}}$ partition, then 
\[
\left\| {(HK)\int_a^b {A_m (t)dt}  - \sum\nolimits_{i = 1}^n {\Delta t_i A_m (\tau _i )} } \right\| 
< {\varepsilon  \mathord{\left/
 {\vphantom {\varepsilon  4}} \right.
 \kern-\nulldelimiterspace} 4}.
\]
Now choose $\delta _2 $ so that, whenever $\{ t_0 ,\tau _1 ,t_1 ,\tau _2 , 
\cdots ,\tau _n ,t_n \}$ is a HK-${{\de_2}}$ partition,
\[
\left\| {(HK)\int_a^b {f_m (t)dt}  - \sum\nolimits_{i = 1}^n {\Delta t_i f_m (\tau _i )} } 
\right\| < {\varepsilon  \mathord{\left/
 {\vphantom {\varepsilon  4}} \right.
 \kern-\nulldelimiterspace} 4}.
\]
Set $\delta  = \delta _1  \wedge \delta _2$ so that, by Lemma 2.19,  $\{ t_0 ,\tau _1 ,t_1 ,\tau _2 , 
\cdots ,\tau _n ,t_n \}$ is a HK-${{\de}_1}$ and HK-${{\de_2}}$ partition so that:
\[
\begin{gathered}
  \left\| {(B)\int_a^b {A(t)dt}  - \sum\limits_{i = 1}^n {\Delta t_i A(\tau _i )} } 
\right\| \leqslant \left\| {(B)\int_a^b {A(t)dt}  - (HK)\int_a^b {A_m (t)dt} } \right\| \hfill \\
  {\text{       }} + \left\| {(HK)\int_a^b {A_m (t)dt}  - \sum\nolimits_{i = 1}^n 
{\Delta t_i A_m (\tau _i )} } \right\| + \left| {(HK)\int_a^b {f_m (t)dt}  
- \sum\nolimits_{i = 1}^n {\Delta t_i f_m (\tau _i )} } \right| \hfill \\
  {\text{                                                                                 
}} + (HK)\int_a^b {f_m (t)dt}  < \varepsilon . \hfill \\ 
\end{gathered} 
\]
\end{proof}
\begin{Def}A function $g: [a,b] \subset \R \to \mcH$ is said to be of bounded variation or BV, 
if (cf. Appendix and \cite{AFP}) 
\[
\mathop {\sup }\limits_\mathbb{P} \left\| {\sum\limits_{i = 1}^n {\left[ {g (b_i ) 
- g (a_i )} \right]} } \right\| < \infty,
\]
for all partitions $\mathbb{P} = \left\{ {(a_1 ,b_1 ), \ldots ,(a_n ,b_n )} \right\}$ 
of the non-overlapping subintervals of $[a,b]$.  In this case, we set 
\[
\mathop {\sup }\limits_\mathbb{P} \left\| {\sum\limits_{i = 1}^n {\left[ {g(b_i ) 
- g(a_i )} \right]} } \right\| = BV_a^b (g).
\]
\end{Def}
\begin{thm} Let $g: [a,b] \to \mcH$, be of BV.
\begin{enumerate}
\item If $h$ is continuous in $[a,b]$, then
\[
I=HK\int_a^b {h(s)dg(s)} 
\]
exists.
\item If, in addition $A$ is a closed densely defined linear operator on $\mcH, \; 
g \in D(A)$ and $Ag(s)=f(s)$ is of BV, then
\begin{equation}
AI=A\int_a^b {h(s)dg(s)}  = \int_a^b {h(s)df(s)} .
\end{equation}
\label{(23)}
\end{enumerate}
\end{thm}
Let $(\Om, \mfB(\Om), \la)$ be a measure space, where $\Om$ is a subset of $\R^n$ 
and $\la$ is the Lebesgue measure.  We know that $L^\iy(\Om, \mfB(\Om), \la)=L^1(\Om, \mfB(\Om), 
\la)^*$, is the dual space of $L^1(\Om, \mfB(\Om), \la)$. For a later study of the Feynman 
integral, we must describe the dual space $L^\iy(\Om, \mfB(\Om), \la)^*$ of  
$L^\iy(\Om, \mfB(\Om), \la)$ in more detail. (Note that the indicator function 
of a set $B,\;  I_B(x)=1, x\in B$ and $0$ otherwise.)
\begin{thm} If $\ell \in L^\iy(\Om, \mfB(\Om), \la)^*$, there is a finitely additive 
complex-valued signed measure $\mu_\ell$ of bounded total variation and 
absolutely continuous with respect to $\la$, such that 
\[
\ell (\phi ) = \int_\Omega  {\phi (x)d\mu_\ell (x),\quad \phi }  \in 
L^\infty  [\Omega ,\mathfrak{B}[\Omega ],\lambda ].
\]
Thus, $L^\iy(\Om, \mfB(\Om), \la)^* = \mfM (\Om, \mfB(\Om), \la)$, 
the finitely additive measures on $\Om$.
\end{thm}
\begin{proof} According to the Jordan Decomposition Theorem, every complex measure 
$\mu$ can be written as $\mu =\mu_1+\mu_2 +i(\mu_3+\mu_4)$, where $\mu_1, \mu_3$  
are positive measures and $\mu_2, \mu_4$ are negative measures (see \cite{GZ}). 
Thus, it is sufficient to prove the theorem when $\mu$ is real. Let $\ell \in  
L^\iy(\Om, \mfB(\Om), \la)^*$ and, for each $B \in \mfB[\Om]$ set 
$\mu_\ell(B) =\ell(I_B)$, where $I_B$ is the indicator function of $B$.  
If $B_1, \, B_2  \in \mfB[\Om], \;B_1 \cap B_2=\emptyset$, then 
$I_{B_1 +B_2}=I_{B_1} +I_{B_2}$ so that 
\[
\ell (I_{B_1 }  + I_{B_2 } ) = \ell (I_{B_1 } ) + \ell (I_{B_2 } ) \Rightarrow 
\mu_\ell \left( {B_1  \cup B_2 } \right) = \mu_\ell \left( {B_1 } \right) 
+ \mu_\ell \left( {B_2 } \right).
\]
Since
\[
\mathop {\sup }\limits_{B \in \mathfrak{B}} \left| {\mu_\ell (B)} \right| 
= \mathop {\sup }\limits_{B \in \mathfrak{B}} \left| {\ell (I_B )} \right| 
\leqslant \left\| \ell  \right\|\left\| {I_B } \right\| < \infty ,
\]
we see that $\mu$ is of bounded variation. 

Let $\phi \in L^\iy(\Om, \mfB(\Om), \la)$ be arbitrary.  For any $\e>0$, 
a simple function $s_\e$ exists such that
\[
s_\varepsilon   = \sum\nolimits_{i = 1}^N {a_i I_{B_i } } ,\quad \lambda \left( 
{B_i  \cap B_j } \right) = 0,\;i \ne j,\quad  \bigcup _{i = 1}^N B_i  = \Omega 
\]
and  
\[
\left\| {\phi  - \sum\nolimits_{i = 1}^N {a_i I_{B_i } } } \right\| < \varepsilon ,
\quad {\text{so that}}\;\left\| {\ell (\phi ) - \sum\nolimits_{i = 1}^N {a_i 
\mu _\ell  (B_i )} } \right\| < \varepsilon \left\| \ell  \right\|.
\]
It follows that
\[
\ell (\phi ) = \int_\Omega  {\phi (x)d\mu _\ell  (x)} ,\quad {\text{and}}\quad \left\| 
\ell  \right\| = \mathop {\sup }\limits_{ess\sup \left| \phi  \right| \leqslant 1} 
\left| {\int_\Omega  {\phi (x)d\mu _\ell  (x)} } \right|.
\]
Finally, because $\la(B)=0$ implies that $I_B=0$ (a.s.), it follows that $\mu_\ell(B)=0$ 
so that $\mu_\ell$ is absolutely continuous with respect to $\la$ (by definition, see \cite{GZ}).
\end{proof}
From here, we see that $L^1(\Om, \mfB(\Om), \la)^{**} = \mfM (\Om, \mfB(\Om), \la)$ and, 
the injection of $L^1(\Om, \mfB(\Om), \la) \to \mfM (\Om, \mfB(\Om), \la)$ is dense.  
Because $L^1(\R^n, \mfB(\R^n), \la)$ is a Banach algebra under convolution, 
it is easy to prove that property.
\begin{cor}  $\mfM(\R^n, \mfB(\R^n), \la)$ is a Banach algebra under convolution.
\end{cor} 
Recall that, if $\Om$ is an open subset of $\R^n$, then  $ \C_c(\Om)$ is the set of all 
continuous functions defined on $\Om$, that vanish outside a compact set.
\begin{cor} If $\phi \in \C_c(\Om)$, then for each $\ell \in \C_c(\Om)^*$, there is 
a countably additive complex measure $\mu_\ell$ such that.
\[
\ell (\phi ) = \int_\Omega  {\phi (x)d\mu _\ell  (x)}.
\]
\end{cor}

\section{Lebesgue Measure on $\R^\iy$}

In this section we briefly review the theory of Lebesgue's measure on 
infinite-dimensional space (see \cite{GZ}). We used a basic requirement for any 
mapping that would serve as an acceptable version of the Lebesgue measure 
on $\R^\iy$. If $I_0= {\left[ { - \tfrac{1}{2},\tfrac{1}{2}} 
\right]^{{\aleph _0}}}$, 
we know that any definition $\la_\iy(\cdot)$ of Lebesgue measure must 
satisfy $\lambda _\infty  \left[ {I_0 } \right] = 1$. 
This requirement is the centerfold of our approach. 

Let $I= [-\tf{1}{2}, \tf{1}{2}], \; I_n=\times_{i = n+1}^\infty I$, so that 
$I_0=\times_{i = 1}^\infty I$ and let $\R_I^n ={\R^n} \times I_n$. It 
is evident that $\mathbb{R}_I^n \subset \mathbb{R}_I^{n+1}$. Since this is 
an increasing sequence, we can define ${\mathbb{R}'}_I^{\infty}$ by: 
\[
{\mathbb{R}'}_I^{\infty} =\rm{lim}_{n\rightarrow \infty}\mathbb{R}_I^n=  
\mathop  \cup \limits_{k =  1}^\infty {\mathbb{R}_I^k}.
\]
Let $\tau_1$ be the topology on ${\mathbb{R}'}_I^{\infty}={\mfX}_1$ induced 
by the class of open sets $\mfO$ defined by:
\[
\mfO=\bigcup_{n=1}^\iy{{\mathfrak{O}}_n} =\bigcup_{n=1}^\iy{ \left\{ {U_n :U_n  
= U \times I_n ,\; U \; {\text{open in }}\mathbb{R}^n } \right\}},
\]
and let $\tau_2$ be the discrete topology on $\R^{\infty} \setminus  
{\mathbb{R}'}_I^{\infty}={\mfX}_2$ induced by the metric $d_2$, for which 
$d_2(x,y)=1, \, x \neq y$  and $d_2(x,y)=0, \, x=y$, for all  $x, y \in {\mfX}_2$.
\begin{Def}We define $({\mathbb{R}}_I^{\infty}, \tau)= \left( {\mfX_1 ,\tau _1 } \right) 
\oplus \left( {\mfX_2 ,\tau _2 } \right)$, the topological direct sum of 
$({\mfX}_1, \tau_1)$ and $({\mfX}_2, \tau_2)$, so that every open set in 
$({\mathbb{R}}_I^{\infty}, \tau)$ is the disjoint union of two sets 
$G_1 \cup G_2$, where $G_1$ is open in $({\mfX}_1, \tau_1)$ and $G_2$ is 
open in $({\mfX}_2, \tau_2)$.
\end{Def}
\begin{Def}  If $A_n=A \times I_n,\; B_n=B \times I_n \in \mathbb{R}_I^n$, we define: 
\begin{enumerate}
\item $A_n \cup B_n= A \cup B \times I_n$,
\item $A_n \cap B_n= A \cap B \times I_n$, and
\item $B_n^c= B^c \times I_n$.
\end{enumerate}
\end{Def}
\begin{Def} 
We define $\mathbb{R}_I^n  = \mathbb{R}^n  \times I_{n} \subset R^\iy$.   
If $T$ is a linear transformation on $\R^n$ and $A_n=A \times I_n$, we define 
the extension operator $T_e$ on $\mathbb{R}_I^n$ by $T_e[A_n]=T[A] \times I_{n}$.

We define the topology on $\mathbb{R}_I^n$  via the following class of open sets:  
\[
{\mathfrak{O}}_n = \left\{ {U \times I_n: ,\;U\;{\rm{open \; in }}\; \mathbb{R}^n } \right\}
\]
and let $\mathfrak{B}[\R_I^n]$ be the natural Borel $\s$-algebra.
\end{Def}
It follows from our construction that $\R_I^{\infty}=\R^{\infty}$ as sets, 
but not as topological spaces. It is easy to prove the following result, which  
shows that convergence in the $\tau$-topology always implies convergence in the Tychonoff topology.
\begin{thm} If $y_k$ converges to $x$ in the $\tau$-topology, then $y_k$ 
converges to $x$ in the Tychonoff topology.
\end{thm}

In a similar manner, if $\mathfrak{B}(\mathbb{R}_I^n )$ is the Borel 
$\s$-algebra for $\mathbb{R}_I^n$ (i.e., the smallest $\s$-algebra generated 
by the ${\mathfrak{O}}_n$), then  $\mathfrak{B}(\mathbb{R}_I^n ) \subset 
\mathfrak{B}(\mathbb{R}_I^{n+1} )$, hence we can define ${\mathfrak{B}'}(\mathbb{R}_I^{\infty} )$ by:
\[
{\mathfrak{B}'}(\mathbb{R}_I^{\infty} )=\rm{lim}_{n\rightarrow \infty} 
\mathfrak{B}(\R_I^{n} )=  \mathop  \cup \limits_{k =  1}^\infty \mathfrak{B}(\R_I^{k} ).
\]
If $\mathcal{P}(A)$ denotes the power set of $A$, let $\mathfrak{B}(\R_I^{\infty})$ 
be the smallest $\sigma$-algebra containing ${ \mathfrak{B}'}(\R_I^{\infty}) \cup    
\mathcal{P} (R_I^{\infty} \setminus \cup_{n=1}^{\infty}\R_I^{n})$. It is evident 
that the class $\mfB(\R_{I}^{\infty})$ coincides with the Borel $\sigma$-algebra 
generated by the $\tau$-topology on $\R^{\infty}$.
For any $A \in \mfB[\R^n]$, we define  $\la_{\infty}(A_n)$ on $\mathbb{R}_I^n$ by product measure:
\[
\la_{\infty}(A_n)=\lambda_n (A) \times \prod _{i = n+1}^\infty  \lambda_1 (I) = \lambda_n (A),
\]

\begin{thm}
$\lambda _\infty$ is a measure on ${\mathfrak{B}}(\R_I^{n})$, equivalent 
to $n$-dimensional Lebesgue measure on $\mathbb{R}^n$.
\end{thm}
\begin{proof} If $A= \mathop  \times \limits_{i =  1}^\infty  A_i \in 
\mathfrak{B}(\mathbb{R}_I^{n})$, then $\la(A_i)=1$ for $i>n$ such that the 
series $\la_\infty  (A) = \prod\nolimits_{i = 1}^\infty  {\lambda (A_i )}$ 
always converges. Furthermore, 
\begin{equation}
0<\la_\infty  (A) = \prod\nolimits_{i = 1}^\infty  {\lambda (A_i )} 
= \prod\nolimits_{i = 1}^n {\lambda (A_i )}=\la_n(\mathop  \times \limits_{i =  1}^n  A_i).
\label{(24)}
\end{equation}
Bearing in mind that sets of the type $A =\mathop \times \limits_{i =  1}^n  A_i$ generate 
$\mathfrak{B}(\mathbb{R}^n )$, we see that $\la_\infty(\cdot)$, restricted to 
$\mathbb{R}_I^{n}$, is equivalent to $\la_n(\cdot)$.
\end{proof}
\begin{cor} The measure $\lambda _\infty  (\cdot)$ is translationally and 
rotationally invariant on $(\mathbb{R}_I^{n},\mathfrak{B}[\mathbb{R}_I^{n} ])$.
\end{cor}

\subsection{Extension of $\la_\iy(\cdot)$ to $\mathbb{R}_I^{\infty}$}

From the previous section, we see that the extension of $\la_\iy(\cdot)$ to 
$\mfB(\R_I^{\infty})$ provides the following important consequences:
\begin{enumerate}
\item $\R_I^n \in \mfB(\R_I^{\infty})$ for all $n$ and $\la_\iy(\cdot)$ is 
equivalent to  $\la_n(\cdot)$ on $\R_I^n$,
\item $\mfB(\R_I^{\infty})$ has a large number of open sets of finite measure and
\item $\la_\iy(I_0)=1$.
\end{enumerate} 
It is shown in \cite{GZ} that $\lambda _\infty  (\cdot)$ can be extended to a 
countably additive measure on ${\mathfrak{B}}(\R_I^{\infty})$ in a constructive 
manner that is closely related to the same approach used for Lebesgue measure on $\R^n$.  
We now consider an equivalent definition of $\la_\iy(\cdot)$ that proves useful, 
in that it is direct and allows us to relate our measure to a number of other 
approaches. We first recall the following characterization of a measure.  
\begin{thm} If $\mcA$ is a $\s$-algebra, the mapping $\mu: \mcA \to [0, \iy]$ 
is a measure if and only if 
\begin{enumerate} 
\item $\mu(\emptyset)=0$.
\item If $A,B \in \mcA$, and $\mu(A\cap B)=0$, then $\mu(A\cup B)=\mu(A)+\mu(B)$.
\item If $\{B_n \} \subset \mcA$ and $B_n \subset B_{n+1}$, then
\[
{\mu }\left( {\bigcup\limits_{k = 1}^\infty  {{B_k}} } \right) = 
\mathop {\lim }\limits_{n \to \infty } {\mu }\left( {{B_n}} \right).
\]
\end{enumerate}
\end{thm} 
\begin{proof} If $\mu$ is a measure, these conditions are necessary.  
Thus, we need only prove that these conditions are sufficient. Since $\mu$ is 
positive and finitely additive, it suffices to show that it is countably 
additive. Let $\{A_n \} \subset \mcA$ be disjoint. From
\[
\bigcup\limits_{k = 1}^n {{A_k}}  \subset \bigcup\limits_{k = 1}^{n + 1} 
{{A_k}} ,\quad  {\text {and}} \quad {\mu }\left( {\bigcup\limits_{k = 1}^n {{A_k}} } 
\right) = \sum\limits_{k = 1}^n {{\mu }\left( {{A_k}} \right)},
\]
if we let 
\[
B_n=\bigcup\limits_{k = 1}^n {{A_k}},
\]
we have $B_n \subset B_{n+1}$, so we can apply condition (3) to get
\begin{equation}
\mathop {\lim }\limits_{n \to \infty } \mu \left( {{B_n}} \right) 
= \mu \left( {\bigcup\limits_{k = 1}^\infty  {{B_k}} } \right) 
= \mu \left( {\bigcup\limits_{k = 1}^\infty  {{A_k}} } \right)
\label{(25)}
\end{equation}
and
\begin{equation}
\quad \mathop {\lim }\limits_{n \to \infty } \mu \left( {{B_n}} \right) 
= \mathop {\lim }\limits_{n \to \infty } \mu\left( {\bigcup\limits_{k = 1}^n 
{{A_k}} } \right) = \mathop {\lim }\limits_{n \to \infty } \sum\limits_{k = 1}^n 
{\mu \left( {{A_k}} \right)}  = \sum\limits_{k = 1}^\infty  {\mu \left( {{A_k}} \right)} .
\label{(26)}
\end{equation}
Combining conditions (2) and (3) proves sufficiency.
\end{proof}
\begin{Def}We define a measure $m_k$ on $\mfB(\R^\iy)$ by
\[
m_k(A)=\la_\iy(A \cap \R_I^k), \; {\tx{for each}} \; A \in \mfB(\R^\iy),
\]
and set
\begin{equation}
m(A) = \mathop {\lim }\limits_{k \to \infty } m _k (A).
\label{(27)}
\end{equation}
\end{Def}
\begin{thm}The mapping $m: \mfB(\R^\iy) \to [0, \iy]$ is a measure. 
\end{thm}
\begin{proof}It is clear that conditions (1) and (2) of Theorem 3.7 are satisfied 
by $m$, thus we need only check (3). Since  $A_i \cap A_j =\emptyset$ unless $i=j$, 
we see that the same is true for $A_i \cap \R_I^k$  and  $A_j \cap \R_I^k$.  
Let $N \in \N$ and note that
\[
\lambda _\infty  \left[ {\left( {\bigcup\limits_{i = 1}^N  {A_i } } \right) 
\cap \mathbb{R}_I^k } \right] = \sum\limits_{i = 1}^N  {\lambda _\infty  
\left( {A_i  \cap \mathbb{R}_I^k } \right)}. 
\]
Since  $A_i  \cap \mathbb{R}_I^k  \subset A_i  \cap \mathbb{R}_I^{k+1}$, 
all the terms are increasing, and therefore
\[
m \left( {\bigcup\limits_{i = 1}^N  {A_i } } \right) = \mathop {\lim }\limits_{k \to \iy } 
m _k \left( {\bigcup\limits_{i = 1}^N  {A_i } } \right) 
= \sum\limits_{i = 1}^N  {\mathop {\lim }\limits_{k \to \infty } m _k 
\left( {A_i } \right)}  = \sum\limits_{i = 1}^N  {m \left( {A_i } \right)} .
\]
Letting $N \to \iy$, completes the proof.
\end{proof}
\begin{cor}The completion $\hat{m}(\cdot)$, of $m(\cdot)$ is equivalent to 
$\la_\iy(\cdot)$ defined on $\mfB(\R_I^\iy)$.
\end{cor}
\begin{thm}$\la_\iy\lt[\R^\iy \setminus \lt(\cup_{n=1}^\iy \R_I^n \rt)\rt]=0$.
\end{thm}
\begin{proof} Using the extension of ${m}_k$ to $\mfB(\R_I^k)$, we have
\[
\lambda _\infty  \left[ {\left( {\mathbb{R}^\infty  \setminus \bigcup
\limits_{k = 1}^\infty  {\mathbb{R}_I^k } } \right)} \right] 
= \mathop {\lim }\limits_{n \to \infty } \hat{m} _n \left[ {\left( {\mathbb{R}^\infty  
\setminus \bigcup\limits_{k = 1}^\infty  {\mathbb{R}_I^k } } \right) 
\cap \mathbb{R}_I^n } \right] = 0.
\]
\end{proof}
\begin{cor} $\la_\iy(\mfX_2)=0$.
\end{cor}
\begin{thm} There exists a family of sets $\{A_k\} \subset \mfB(\R_I^\iy)$ 
with $\la_\iy(A_k) < \iy$, and a set $N$ with $\la_\iy(N)=0$ such that:
\begin{equation}
\R_I^\iy = {\cup}_{k=1}^\iy A_k \cup N.
\label{(28)}
\end{equation}
\end{thm}
\begin{proof} Because $\la_\iy(\cdot)$ is concentrated on $\mfX_1$, we  
set $N= \mfX_2$.  To show  that (28) holds, let  $\{x_k\}$ be the set of 
vectors in $\R^\iy$ with rational coordinates and let $B_k$ be the unit cube 
with center $x_k$, so that $\la_\iy(B_k)=1$ for all $k$ and 
\[
\R_I^\iy = {\cup}_{k=1}^\iy B_k.
\]
Let $A_k=B_k \setminus \mfX_2$.
\end{proof}

\section{{Lebesgue Measure on Hilbert Space}} 

In this section we construct the Lebesgue measure on a separable Hilbert space.   
Let $\mcH$ be a separable Hilbert space and let $\{e_n\}$ be an orthogonal 
basis for $\mcH$ such that for any $x \in \mcH$ there is a sequence  of 
scalars $(x_n)_{n=1}^\iy$, such that $x = \sum\nolimits_{n = 1}^\infty  {{x_n}} {e_n}$ and 
${\sum\nolimits_{n = 1}^\infty  {\left| {{x_n}} \right|} ^2} < \infty$
\begin{Def} We define $\mcH_I^n$ and $\mcH_I$ by:
\[
{\mcH_I^n} = \left\{ {\left( {{x_k}} \right)_{k = 1}^n :x = 
\sum\limits_{k = 1}^n  {{x_k}{e_k}}  \in \mcH} \right\}\times I_n
\]
and 
\[
{\mcH_I} = \left\{ {\left( {{x_k}} \right)_{k = 1}^\infty :x = 
\sum\limits_{n = 1}^\infty  {{x_k}{e_k}}  \in \mcH} \right\}.
\]
\end{Def} 
It can be seen that $\mcH_I^n=\R_I^n$. This observation will prove important 
when we discuss integration in the next section. 
\begin{Def}Define $T_n:\mcH \to \mcH_I^n$ by $T_n(x) = \left( {x_k } 
\right)_{k = 1}^n$ and $T:\mcH \to \mcH_I ,\;T(x) 
= \left( {x_k } \right)_{k = 1}^\infty$. Define  
\[
 \mathfrak{B}\left( {{\mcH_I}} \right) = {\mcH_I} \cap 
\mathfrak{B}\left( {\mathbb{R}_I^\infty } \right)
\]
and define
\[
\mathfrak{B}_I [\mcH] = \left\{ {T^{ - 1} (A)\left| {\; A \in 
\mathfrak{B}[\mcH_I ]} \right.} \right\} = :T^{ - 1} \left\{ 
{\mathfrak{B}\left[ {\mcH_I } \right]} \right\}.
\]
\end{Def}
\begin{Def} Let $J_k =[-\tf{1}{2k^{3/2}}, \tf{1}{2k^{3/2}}]$ and $g_k(x_k)
=k^{3/2}\chi_{J_k}(x_k)$. For each $n$, we define $\nu_n$ by:
\[
{\nu _n} = \mathop  \otimes \limits_{k = 1}^n \lambda  \otimes \left( {\mathop  
\otimes \limits_{k = n+1}^\infty  {\mu _k}} \right) = {\lambda _n} \otimes 
\left( {\mathop  \otimes \limits_{k = n+1}^\infty  {\mu _k}} \right),{\text{ where}}\;\; 
d{\mu _k} = g_k(x_k)d\lambda (x_k)
\]
and for each $A \in \mfB_I(\mcH)$,  we define $\la_\mcH(A)$ by:
\[
{\lambda _\mcH}(A) = \mathop {\lim}\limits\sup_{n \to \infty } {\nu _n}[T(A) \cap \R_I^n].
\]
\end{Def} 
\begin{thm}The measure ${\la}_{{\mcH}}$ is a nonvanishing 
$\sigma$-finite Borel measure on $\mcH$. 
\end{thm}
\begin{proof} The result follows if we show that $\la_\mcH$ is nonzero 
and countably additive. Since
\[
\left\| {\sum\limits_{k = 1}^\infty  {\frac{{e_k }}
{{k^{3/2} }}} } \right\|_\mcH  \leqslant \sum\limits_{k = 1}^\infty {\frac{1}{{k^{3/2} }}}  < \infty ,
\]
we see that the set $A = \left\{{  \sum_{k = 1}^\infty{x_k e_k}  :\left| {x_k } 
\right| \leqslant \tfrac{1}{{k^{3/2} }}} \right\} \in \mathfrak{B}_I\left( {\mcH } \right)$, so that
\[
{\lambda _\mcH}(A) \geqslant {\nu _1}(T[A]) = \lambda ({I_1})\prod\limits_{k = 2}^\infty  
{{\mu _k}({J_k}) = 1}. 
\] 
We now note that ${\nu _1}[T(A)] \geqslant {\nu _n}(T[A])$ for all $n$ and, because 
$\{\nu_n \}$ is an increasing sequence, we see that ${\lambda _\mcH}(A)=1$.

To see that ${\lambda _\mcH}$ is countably additive, let $\{A_k \}$ be disjoint and 
let $B_k=\cup_{i = 1}^k {{A_i}}$.  Because the family $\{B_k \}$ increases for each $n$, 
\[
\nu _n \left[ {T\left( {\bigcup\limits_{i = 1}^\infty  {A_i } } \right) \cap 
\mathbb{R}_I^n } \right] = \mathop {\lim }\limits_{k \to \infty } \nu _n 
\left[ {T(B_k ) \cap \mathbb{R}_I^n } \right] = \sum\limits_{i = 1}^\infty  
{\nu _n \left[ {T\left( {A_i } \right) \cap \mathbb{R}_I^n } \right]},
\]
so that
\[
\lambda _\mcH \left[ {\bigcup\limits_{i = 1}^\infty  {A_i } } \right] 
= \sum\limits_{i = 1}^\infty  {\mathop {\lim }\limits\sup_{n \to \infty } 
\nu _n \left[ {T\left( {A_i } \right) \cap \mathbb{R}_I^n } \right]}  
= \sum\limits_{i = 1}^\infty  {\lambda _\mcH \left[ {A_i } \right]} .
\]
\end{proof}
\begin{Def} We call $\mcH_I$ the canonical representation of $\mcH$ in $\R_I^\iy$.
\end{Def}
The following theorem (proved in \cite{GZ}), characterizes the properties of the 
Lebesgue measure on $\R_I^\iy$ (or $\mcH$). We state the results for $\R_I^\iy$.  
\begin{thm}{\lb{Y}} The measure space $(\R_I^\iy, \ \mfB[\R_I^\iy],\ 
\la_\iy)$ has the following properties:
\begin{enumerate} 
\item $\la_\iy(\mfX_2)=0$.
\item For every $A \in \mfL[\R_I^\iy]$ (Lebesgue sets) and $\e>0$, there 
exists a compact set $F \subset A$ and an open set $G \supset A$ 
such that  $\la_\iy(G \setminus F)<\e$, so that $\la_\iy(\cdot)$ is regular.
\item There exists a family of compact sets $\{A_n\} \subset \mfB[\R_I^\iy]$, 
with $\la_\iy[A_n] <\iy$ and a set $N$ with $\la_\iy[N]=0$, such that 
$\mathbb{R}_I^\infty  = \bigcup\nolimits_{n = 1}^\infty  {{A_n} \cup N}$ 
(i.e., $\la_\iy(\cdot)$ is $\s$-finite).
\item For $A \in \mfB[\R_I^\iy], \; \la_\iy(A-x)=\la_\iy(A)$ if and 
only if  $x \in {\ell _1}$. 
\end{enumerate}
\end{thm} 

\subsection{\bf{Measurable Functions}} 

In this subsection we discuss measurable functions on $\R_I^\iy$ (or $\mcH$). 
Let $x=(x_1,x_2,x_3, \dots) \in \R_I^\iy$. Fixing $n$ with 
$
I_n=\prod_{k=n+1}^{\iy}{\lt[-\tf{1}{2},\tf{1}{2}\rt]},
$
we set ${h_n}(\hat{x}) ={\chi _{I_n}}({\hat{x}})$ (indicator function 
of $I_n$), where $\hat{x}= (x_i)_{i={n+1}}^{\iy}$.     
\begin{Def} Let $\mcM^n$ represent the class of Lebesgue measurable 
functions on  $\mathbb{R}^n$. If $x \in \R_I^\iy$ and  $f^n \in \mcM^n$, 
let $\bar{x}= (x_i)_{i={1}}^{n}$ and define an essentially tame measurable 
function of order $n$ (or $e_n$-tame) on $\R_I^\iy$ by
$f(x)=f^n(\bar{x}) \otimes{h}_n(\hat{x}) $. We let   
\[
{\mathcal{M}}_I^n  = \left\{ {f(x): f(x) = f^n(\bar{x}) 
\otimes{h}_n(\hat{x}), \; x \in \R_I^\iy}  \right\} 
\]
be the class of all $e_n$-tame functions.
\end{Def} 
\begin{Def}\lb{m} A function $f: \mathbb{R}_I^{\iy} \to \R$ is said to be 
measurable and we write $f \in \mcM_I$, if there is a sequence 
$\{  f_n \in \mcM_I^{n} \}$ of $e_n$-tame functions, such that 
$\mathop{\lim}_{n \to \iy}f_n(x) \to f(x)$  $\la_\iy$-(a.e).
\end{Def}
The existence of functions satisfying Definition \rf{m} is not obvious, so we have \cite{GM}:
\begin{thm}
\rm(Existence)\it\ Suppose that $f:\R^{\iy}_I\to(-\iy,\iy)$\ and 
$f^{-1}(A)\in\mfB[\R^{\iy}_I]$\ for all $A\in\mfB[\R]$. Then there 
exists a family of functions $\{f_n\}$, $f_n\in\mcM_I^n$, such that 
$f_n(x)\to f(x)$, $\la_{\iy}$-\rm(a.e).\it
\end{thm}
\begin{rem} From Theorem {\rf{Y}} (1), we can see that any set $A$, of nonzero 
measure is concentrated in $\mfX_1$ (i.e., $\la_\iy(A)=\la_\iy(A \cap 
\mfX_1)$). It also follows that the essential support of the limit function 
$f(x)$ in definition 2.14 (i.e., $\left\{ {x\left| {f(x) \ne 0} \right.} 
\right\}$), is concentrated in $\R_I^N$, for some $N$.
\end{rem}

\subsection{Integration Theory on $\R_I^\iy$ or $\mcH$}

In this section all theorems remain true when $\R_I^\iy$ is replaced 
by $\mcH$, hence it is sufficient to  provide all results for $\R_I^\iy$.  

We provide a constructive theory of integration on 
$\R_I^\iy$ using the known properties of integration on $\R_I^n$.   
This approach has the advantage that all standard theorems for 
the Lebesgue measure apply. (The proofs are the same as for integration on $\R^n$.)  
Let ${{L}}^1 [\R_I^{n}]$ be the class of integrable functions on 
$\R_I^{n}$. Since ${{L}}^1 [\R_I^n] \subset {{L}}^1 [\R_I^{n+1}]$, we 
define ${{L}}^1 [{\hat{\R}}_I^{\infty}]=\bigcup_{n=1}^{\iy}{{L}}^1[\R_I^n]$.  
\begin{Def}
We say that a measurable function $f \in {{L}}^1 [\R_I^{\infty}]$, if there 
is a Cauchy sequence  $\{f_n \} \subset {{L}}^1 [{\hat{\R}}_I^{\iy}]$, with 
$f_n \in {{L}}^1 [{{\R}}_I^{n}]$ and ${\lim _{n \to \infty }}{ {{f_n}(x) 
= f(x) }  },\; \la_\iy$-\rm(a.e). 
\end{Def}
\begin{thm} ${{L}}^1 [{\hat{\R}}_I^{\infty}]=L^1[\R_I^\iy]$.
\end{thm}
\begin{proof} We know that ${{L}}^1 [{\hat{\R}}_I^{\infty}] \supset 
L^1[\R_I^n]$, for all $n$, hence it suffices to prove that 
${{L}}^1 [{\hat{\R}}_I^{\infty}]$ is closed. Let $f$ be a limit point of 
${{L}}^1 [{\hat{\R}}_I^{\infty}]$ ($f \in L^1[\R_I^\iy]$). If $f=0$, we are 
done, thus assume $f \ne 0$. From the remarks above, we know 
that if $A_f$ is the support of $f$, then 
$\la_\iy(A_f)=\la_\iy(A_f \cap \mfX_1)$. Thus, $A_f\cap \mfX_1 \subset  
\R_I^N$ for some $N$. This implies that there is a function 
$f' \in {{L}}^1 [{{\R}}_I^{N+1}]$, with  
${\lambda _\infty }\left( {\left\{ {{\bf{x}}:f({\bf{x}}) \ne {f'}({\bf{x}})} 
\right\}} \right) = 0$. It follows that, $f({\bf{x}}) ={f'}({\bf{x}})$ (a.e). 
Recalling that $L^1[\R_I^n]$ is a set of equivalence classes, we observe 
that ${{L}}^1 [{\hat{\R}}_I^{\infty}]=L^1[\R_I^\iy]$.
\end{proof}
 
\begin{Def}If $f \in L^1[\R_I^\iy]$, we define the integral of $f$ by
\[
\int_{\mathbb{R}_I^\infty } {f(x)d{\lambda _\infty }(x)}  = \mathop 
{\lim }\limits_{n \to \infty } \int_{\mathbb{R}_I^\iy } {{f_n}(x)d{\lambda _\infty }(x)},
\] 
where $\{f_n\} \subset {{L}}^1 [{{\R}}_I^{\iy}]$ is any Cauchy sequence 
converging to $f(x)$-\rm(a.e). 
\end{Def}
\begin{thm}{\lb{God1}} If $f \in L^1[\R_I^\iy]$, then the above integral 
exists and all theorems that are true for $f \in L^1[\R_I^n]$, 
also hold for $f \in L^1[\R_I^\iy]$. 
\end{thm}

\section{The Kuelbs-Steadman Spaces}

If we want a class of spaces designed to include the HK-integrable 
functions on $\R_I^n$, for a fixed $n$, let $\mathbb{Q}_I^n$ be the set 
$\left\{ {{\mathbf{x}} = (x_1 ,x_2  \cdots  ) \in {\mathbb{R}}_I^n } \right\}$ 
such that the first $n$ coordinates $(x_1 ,x_2  \cdots x_n )$ are 
rational. Since every countable set is isomorphic to the set of natural numbers,  
$\mathbb{Q}_I^n$ is a countable set which is also dense  in ${\mathbb{R}}_I^n $.  Therefore,  we 
can arrange it as $\mathbb{Q}_I^n  = \left\{ {{\mathbf{x}}_1, {\mathbf{x}}_2, 
{\mathbf{x}}_3,  \cdots } \right\}$. For each $l$ and $i$, let 
${\mathbf{B}}_l ({\mathbf{x}}_i ) $ be the closed cube centered at 
${\mathbf{x}}_i$, with sides parallel to the coordinate axes and edge 
$e_l=\tfrac{1 }{2^{l-1}\sqrt{n}}, l \in {\mathbb{N}}$. Now choose 
the natural order that maps $\mathbb{N} \times \mathbb{N}$ bijectively to 
$\mathbb{N}$, and let $\left\{ {{\mathbf{B}}_{k} ,\;k \in 
\mathbb{N}}\right\}$ be the resulting set of (all) closed cubes 
$\{ {\mathbf{B}}_l ({\mathbf{x}}_i )\;\left| {(l,i) \in \mathbb{N} 
\times \mathbb{N}\} } \right.
$
centered at a point in $\mathbb{Q}_I^n $. Let ${\mathcal{E}}_k ({\mathbf{x}})$ 
be the characteristic function of ${\mathbf{B}}_k $, so that 
${\mathcal{E}}_k ({\mathbf{x}})$ is in ${{L}}^p [{\mathbb{R}}_I^n ] \cap 
{{L}}^\infty  [{\mathbb{R}}_I^n ] $ for $1 \le p < \infty$.  
We define $F_{k} (\; \cdot \;)$ on ${{L}}^1 [{\mathbb{R}}_I^n ]$ as
\begin{equation}
{F_k}(f) = \int_{\R_I^n} {{\mcE_k}({\bf{x}})f({\bf{x}})d{\lambda _\infty}
({\bf{x}})}  = \int_{{{\bf{B}}_k}} {f({\bf{x}})d{\lambda _\infty }({\bf{x}})} .
\label{(29)}
\end{equation}
Since each ${\mathbf{B}}_{k}$ is a cube with sides parallel to the coordinate 
axes, $F_{k} (\; \cdot \;)$ is defined for all HK-integrable functions, and  
is a bounded linear functional on ${{L}}^p [{\mathbb{R}}_I^n ] $ for each 
${k}$, with $\left\| {F_{k} } \right\|_\infty   \le 1$. Moreover, if $F_k (f) = 0$ 
for all ${k}$, then $f = 0$, thus $\left\{ {F_{k} } \right\}$ is fundamental 
on ${{L}}^p [{\mathbb{R}}_I^n ] $ for $1 \le p < \infty$.
Fix ${t_{k}}> 0$ such that ${\sum\nolimits_{k = 1}^\infty  {t_k}}=1$  
and define a new inner product $\left( {\; \cdot \;} \right)$ 
on ${{L}}^2 [{\mathbb{R}}_I^n ]$ by
\begin{equation}
\left( {f,g} \right)  = \sum\nolimits_{k = 1}^\infty  {t_k } 
\left[ {\int_{\mathbb{R}_I^n } {{\mathcal{E}}_k ({\mathbf{x}})f({\mathbf{x}})
d\la_\iy({\mathbf{x}})} } \right] \overline{\left[ {\int_{\mathbb{R}_I^n } 
{{\mathcal{E}}_k ({\mathbf{y}})g({\mathbf{y}})d\la_\iy({\mathbf{y}})} } \right]}.   
\label{(30)}
\end{equation}
The completion  
of ${{L}}^2 [{\mathbb{R}}_I^n ]$ in this inner product 
is the Kuelbs-Steadman space, ${K}{S}^2 [{\mathbb{R}}_I^n ]$.   
To see directly that ${K}{S}^2 [{\mathbb{R}}_I^n ]$ contains the 
HK integrable functions, let $f$ be HK integrable,
\[
\left\| f \right\|_{{KS}^2}^2  = \sum\nolimits_{k = 1}^\infty  {t_k } 
\left| {\int_{\mathbb{R}_I^n } {{\mathcal{E}}_k ({\mathbf{x}})
f({\mathbf{x}})d\la_\iy({\mathbf{x}})} } \right|^2  \leqslant \sup _k 
\left| {\int_{{\mathbf{B}}_{k} } {f({\mathbf{x}})d\la_\iy({\mathbf{x}})} } \right|^2  <\iy, 
\]
hence $f \in {K}{S}^2 [{\mathbb{R}}_I^n ] $.  
\begin{thm} For each $p,\;1 \leqslant p \leqslant \infty,$
 ${K}{S}^2 [{\mathbb{R}}_I^n ] \supset {{L}}^p [{\mathbb{R}}_I^n ]$
as a continuous dense subspace.
\end{thm}
\begin{proof} By construction, ${K}{S}^2 [{\mathbb{R}}_I^n ]$
contains ${{L}}^2 [{\mathbb{R}}_I^n ]$ densely thus, we only need to show that 	
${K}{S}^2 [{\mathbb{R}}_I^n ] \supset {{L}}^q [{\mathbb{R}}_I^n ]$
for $q \ne 2$.  If $f \in {{L}}^q [{\mathbb{R}}_I^n ]$ and $q < \infty $, then we have 
\[
\begin{gathered}
\left\| f \right\|_{{KS}^2}  = \left[ {\sum\nolimits_{k = 1}^\infty  
{t_k } \left| {\int_{{\mathbb{R}}_I^n } {{\mathcal{E}}_k ({\mathbf{x}})
f({\mathbf{x}})d\la_\iy({\mathbf{x}})} } \right|^{\frac{{2q}}{q}} } 
\right]^{1/2}  \hfill \\
{\text{ }} \leqslant \left[ {\sum\nolimits_{k = 1}^\infty  {t_k } 
\left( {\int_{{\mathbb{R}}_I^n } {{\mathcal{E}}_k ({\mathbf{x}})\left| 
{f({\mathbf{x}})} \right|^q d\la_\iy({\mathbf{x}})} } \right)^{\frac{2}{q}} } 
\right]^{1/2}  \hfill \\
{\text{      }} \leqslant \sup _k \left( {\int_{{\mathbb{R}}_I^n } 
{{\mathcal{E}}_k ({\mathbf{x}})\left| {f({\mathbf{x}})} \right|^q 
d\la_\iy({\mathbf{x}})} } \right)^{\frac{1}
{q}}  \leqslant \left\| f \right\|_q . \hfill \\ 
\end{gathered} 
\]
Thus, $f \in {K}{S}^2 [{\mathbb{R}}_I^n ] $.  For $q = \infty $, 
first note that $ vol({\mathbf{B}}_k )^2 \le \left[ {\frac{1}
{{\sqrt n }}} \right]^{2n} \le 1$, hence we have 
\[
\begin{gathered}
\left\| f \right\|_{{KS}^2}  = \left[ {\sum\nolimits_{k = 1}^\infty  
{t_k } \left| {\int_{{\mathbf{R}}_I^n } {{\mathcal{E}}_k ({\mathbf{x}})
f({\mathbf{x}})d\la_\iy({\mathbf{x}})} } \right|^2 } \right]^{1/2}  \hfill \\
{\text{       }} \leqslant \left[ {\left[ {\sum\nolimits_{k = 1}^\infty  
{t_k [vol({\mathbf{B}}_k )]^2 } } \right][ess\sup \left| f \right|]^2 } \right]^{1/2}  
\leqslant \left\| f \right\|_\infty  . \hfill \\ 
\end{gathered} 
\]
Therefore, $f \in {K}{S}^2 [{\mathbb{R}}_I^n ]$ and ${{L}}^\infty  [{\mathbb{R}}_I^n ] 
\subset {K}{S}^2 [{\mathbb{R}}_I^n ]$.
\end{proof}
Before proceeding to an additional discussion, we construct the 
${K}{S}^p [{\mathbb{R}}_I^n ]$ spaces, for $1 \le p \le \iy$. Define:
\[
\left\| f \right\|_{{{KS}}^p }  = \left\{ {\begin{array}{*{20}c}
{\left\{ {\sum\nolimits_{k = 1}^\infty  {t_k \left| {\int_{\mathbb{R}_I^n } 
{ {\mathcal{E}}_k ({\mathbf{x}})f({\mathbf{x}})d\la_\iy({\mathbf{x}})} } 
\right|} ^p } \right\}^{1/p} ,1 \leqslant p < \infty},  \\
   {\sup _{k \geqslant 1} \left| {\int_{\mathbb{R}_I^n } { {\mathcal{E}}_k 
({\mathbf{x}})f({\mathbf{x}})d\la_\iy({\mathbf{x}})} } \right|,p = \infty .}  \\

 \end{array} } \right.
\] 
It is easy to see that $\left\| \cdot \right\|_{{{KS}}^p }$ defines a norm 
on ${L}^p$. If ${{{KS}}^p }$ is the completion of ${L}^p$ with respect 
to this norm, we have:
\begin{thm} For each $q,\;1 \leqslant q \leqslant \infty,$
${K}{S}^p [{\mathbb{R}}_I^n ] \supset {{L}}^q [{\mathbb{R}}_I^n ]$ 
as a dense continuous embedding.
\end{thm}
\begin{proof} As in the previous theorem, by construction ${K}{S}^p [{\mathbb{R}}_I^n ]$
contains ${{L}}^p [{\mathbb{R}}_I^n ]$ densely, hence we only need to show that 	
${K}{S}^p [{\mathbb{R}}_I^n ] \supset {{L}}^q [{\mathbb{R}}_I^n ]$
for $q \ne p$. First, suppose that $p< \infty$.  If $f \in {{L}}^q 
[{\mathbb{R}}_I^n ]$ and $q < \infty $, then we have 
\[
\begin{gathered}
\left\| f \right\|_{{KS}^p}  = \left[ {\sum\nolimits_{k = 1}^\infty  
{t_k } \left| {\int_{{\mathbb{R}}_I^n } {{\mathcal{E}}_k ({\mathbf{x}})
f({\mathbf{x}})d\la_\iy({\mathbf{x}})} } \right|^{\frac{{qp}}{q}} } 
\right]^{1/p}  \hfill \\
{\text{  }} \leqslant \left[ {\sum\nolimits_{k = 1}^\infty  {t_k } 
\left( {\int_{{\mathbb{R}}_I^n } {{\mathcal{E}}_k ({\mathbf{x}})
\left| {f({\mathbf{x}})} \right|^q d\la_\iy({\mathbf{x}})} } 
\right)^{\frac{p}{q}} } \right]^{1/p}  \hfill \\
{\text{      }} \leqslant \sup _k \left( {\int_{{\mathbb{R}}_I^n } 
{{\mathcal{E}}_k ({\mathbf{x}})\left| {f({\mathbf{x}})} \right|^q 
d\la_\iy({\mathbf{x}})} } \right)^{\frac{1}
{q}}  \leqslant \left\| f \right\|_q . \hfill \\ 
\end{gathered} 
\]
Hence, $f \in {K}{S}^p [{\mathbb{R}}_I^n ] $.  For $q = \infty $, we have 
\[
\begin{gathered}
  \left\| f \right\|_{{KS}^p}  = \left[ {\sum\nolimits_{k = 1}^\infty  
{t_k } \left| {\int_{{\mathbb{R}}_I^n } {{\mathcal{E}}_k ({\mathbf{x}})
f({\mathbf{x}})d\la_\iy({\mathbf{x}})} } \right|^p } \right]^{1/p}  \hfill \\
  {\text{       }} \leqslant \left[ {\left[ {\sum\nolimits_{k = 1}^\infty  
{t_k [vol({\mathbf{B}}_k )]^p } } \right][ess\sup \left| f \right|]^p } 
\right]^{1/p}  \leqslant M\left\| f \right\|_\infty  . \hfill \\ 
\end{gathered} 
\]
Thus, $f \in {K}{S}^p [{\mathbb{R}}_I^n ]$ and ${{L}}^\infty  
[{\mathbb{R}}_I^n ] \subset {K}{S}^p [{\mathbb{R}}_I^n ]$.   
The case $p=\infty$ is obvious.
\end{proof} 
\begin{thm} For ${K}{S}^p$, $1\leq p \leq \infty$, we have: 
\begin{enumerate}
\item If $f,g \in {K}{S}^p$, then 
$
\left\| {f + g} \right\|_{{{KS}}^{{p}} }  \leqslant \left\| f 
\right\|_{{{KS}}^{{p}} }  + \left\| g \right\|_{{{KS}}^{{p}} }
$  (Minkowski inequality). 
\item If $K$ is a weakly compact subset of ${{L}^p}$, then it is a compact 
subset of  ${K}{S}^p$.
\item If $1< p < \infty$, then ${K}{S}^p$ is uniformly convex.
\item If $1< p < \infty$ and $p^{ - 1}  + q^{ - 1}  = 1$, the dual 
space of ${K}{S}^p$ is ${K}{S}^q$.
\item  ${K}{S}^{\infty} \subset {K}{S}^p$, for $1\leq p < \infty$.
\end{enumerate}
\end{thm}
\begin{thm} For each $p, \; 1 \le p \le \infty$, the test functions 
$\mcD \subset {KS}^p (\mathbb{R}_I^n)$ as a continuous embedding.
\end{thm}
\begin{proof} Because $KS^\iy(\R_I^n)$ is continuously embedded in 
${KS}^p (\mathbb{R}_I^n), \; 1 \le q < \infty$, it suffices to prove 
the result for $KS^\iy(\R_I^n)$. Suppose that $\phi_j \to \phi$ in 
$\mcD[\R_I^n]$, so that there exists a compact set $K \subset \R_I^n$, 
containing the support of $\ph_j -\phi$ and ${D^\alpha }{\phi_j}$ 
converges to ${D^\alpha }\phi $ uniformly on $K$ for every multi-index 
$\al$.  Let $L= \{l \in \N: {\tx{the support of }\mcE_l}, 
\ stp\{\mcE_l\} \subset K \}$, then 
\[
\begin{gathered}
  \mathop {\lim }\limits_{j \to \infty } {\left\| {{D^\alpha }\phi  - 
{D^\alpha }{\phi _j}} \right\|_{KS}} = \mathop {\lim }\limits_{j \to \infty } 
\mathop {\sup }\limits_{l \in L} \left| {\int_{{\mathbb{R}_I^n}} 
{\left[ {{D^\alpha }\phi \left( x \right) - {D^\alpha }{\phi _j}\left( 
x \right)} \right]{{\mathcal{E}}_l}\left( x \right)d{\lambda _\iy}\left( 
x \right)} } \right| \hfill \\
\le vol({{\mathbf{B}}_{l}})   \mathop {\lim }\limits_{j \to \infty } 
\mathop {\sup }\limits_{x \in K} \left| {{D^\alpha }\phi \left( x \right) - 
{D^\alpha }{\phi _j}\left( x \right)} \right| \leqslant   \mathop {\lim }
\limits_{j \to \infty } \mathop {\sup }\limits_{x \in K} \left| 
{{D^\alpha }\phi \left( x \right) - {D^\alpha }{\phi _j}\left( 
x \right)} \right| =0. \hfill \\ 
\end{gathered} 
\]
It follows that $\mcD[\mathbb{R}_I^n] \subset {KS}^p [\mathbb{R}_I^n]$ 
as a continuous embedding, for $1 \le p \le \iy$. Thus, using the 
Hahn-Banach theorem, we see that the Schwartz distributions, 
$\mcD^*[\mathbb{R}_I^n] \subset [{KS}^p (\mathbb{R}_I^n)]^* $, for $1 \le p \le \iy$.
\end{proof}

\subsection{The Family $KS^p[\R_I^\iy]$}

We can now construct the spaces $KS^p[\R_I^\iy], \; 1 \le p \le \iy$, using 
the same approach that led to ${{L}}^1 [{{\R}}_I^{\infty}]$.  
Since ${{KS}}^p [\R_I^n] \subset {{KS}}^p [\R_I^{n+1}]$, we define 
${{KS}}^p [{\hat{\R}}_I^{\iy}]=\bigcup_{n=1}^{\iy}{{KS}}^p[\R_I^n]$.  
\begin{Def}
We say that a measurable function $f \in {{KS}}^p [\R_I^{\infty}]$, if 
there is a Cauchy sequence $\{f_n \} \subset {{KS}}^p [{\hat{\R}}_I^{\iy}]$, 
with $f_n \in {{KS}}^p [{{\R}}_I^{n}]$ and ${\lim _{n \to \infty }}
{{{f_n}(x) = f(x) }  },\; \la_\iy$-\rm(a.e). 
\end{Def} 
The same proof as Theorem 1.18 shows that functions in 
${{KS}}^p [{\hat{\R}}_I^{\iy}]$ differ from functions in its closure 
${{KS}}^p [{{\R}}_I^{\iy}]$, by sets of measure zero.
\begin{thm} ${{KS}}^p [{\hat{\R}}_I^{\infty}]={{KS}}^p[\R_I^\iy]$.
\end{thm}
\begin{Def}If $f \in {{KS}}^p[\R_I^\iy]$, we define the integral of $f$ by:
\[
\int_{\mathbb{R}_I^\infty } {f(x)d{\lambda _\infty }(x)}  = \mathop 
{\lim }\limits_{n \to \infty } \int_{\mathbb{R}_I^\iy } {{f_n}(x)d{\lambda _\infty }(x)},
\] 
where $f_n \in {{KS}}^p [{{\R}}_I^{n}]$ is any Cauchy sequence converging to $f(x)$. 
\end{Def}
\begin{thm}{\lb{God2}} If $f \in {{KS}}^p[\R_I^\iy]$, then the above 
integral exists and all theorems that are true for $f \in {{KS}}^p[\R_I^n]$, 
also hold for $f \in {{KS}}^p[\R_I^\iy]$. 
\end{thm}
\begin{rem}Theorem \rf{God2} is true if $(\R_I^\iy, \; {\lambda _\infty })$  
is replaced by $(\mcH, \; {\lambda _\mcH })$.
\end{rem}

\section{Feynman Operator Calculus}

In this section we describe the construction of Feynman's time ordered operator calculus. 
We first prove that a unique definition of time exists based on an assumption that is
weaker than the cosmological principle (i.e., the universe is homogenous and 
isotropic in the large). If $O$ and $O'$ are any two observers of a particle with 
velocity ${\mathbf{v}}$ (${\mathbf{v}'}$) and $\gamma^{-1} \left( {\mathbf{v}} \right) 
= \sqrt {1 - \tfrac{{{ {\mathbf{v}}^2} }}{{{c^2}}}}, \; {\mathbf{v}}=\tf{d{\bf{x}}}{dt} $,  
the standard definition of proper time is: 
\begin{equation}\lb{p1}
d\tau  = {\gamma ^{ - 1}}\left( {\mathbf{v}} \right)dt,\quad d\tau  
= {\gamma ^{ - 1}}\left( {{\mathbf{v}'}} \right)dt'.
\end{equation}
Using the relations $H =\gamma \left( {\mathbf{v}} \right)mc^2$ for observer 
$O$ and $H' =\gamma \left( {\mathbf{v}'} \right)mc^2$ for observer $O'$, we obtain
\begin{equation}\lb{p3}
d\tau  = \tfrac{{m{c^2}}}{H}dt,\quad d\tau  = \tfrac{{m{c^2}}}{{H'}}dt'.
\end{equation}
\begin{thm}\lb{p3t} If the universe is representable, in the sense that 
the ratio of the observed total mass energy to the total energy is 
independent of any observer's particular portion of the universe, then 
the universe has a unique clock.
\end{thm}
\begin{proof} Let $O$ and $O'$ be two observers viewing the universe.
If $Mc^2$ is the total mass energy and $H$ is the total energy Hamiltonian for 
the universe as seen by $O$, under the stated conditions $H/{Mc^2}$ is constant.  
By assumption this view is observer independent, so that $O'$  will also obtain the 
same ratio. It follows from (\rf{p3}) that
\[  
 d{\tau} = \frac{{M{c^2}}}{H}dt=\frac{{M{c^2}}}{H}dt'= d\tau_h.  
 \]
Thus, $\tau_h= \frac{{M{c^2}}}{H}t$ uniquely defines a universal clock for any observer.
\end{proof}  
The fifth parameter was first introduced by Fock, and later by Stueckelberg followed by 
Feynman and Schwinger. They each considered it a global clock in addition to the Minkowski 
spacetime. Horwitz and Piron \cite{HP}, and Fanchi and Collins named it  
historical time. They were the first to realize that there should be physical 
justification for this clock (see \cite{FA}). 

\subsection{Feynman-Dyson Space}

Let $\mathcal{H}=KS^2[\R_I^n]$, with a fixed orthonormal basis 
$\left\{ e^i  \right\}$. Let  $J = [ - T,T]$ be an interval of historical 
time and, for each  $t \in J$, define $\mcH(t) =\mathcal{H}$ and let ${\mcH_ \otimes } 
= {\hat  \otimes _t}\mcH(t)$ be the continuous tensor product Hilbert space of 
von Neumann over $J$.  For each $i \in \N$, set $e_t^i = {e^i}$ and  
${ {{E}}^i} = { \otimes _{t}}e_t^i$. We  define ${\mathcal{F}}{ {\mathcal{D}}^i} 
\subset {\mathcal{H}_ \otimes }$ to be the smallest Hilbert space containing  
${ {{E}}^i}$. We call ${\mathcal{F}\mathcal{D}}=  \oplus _{i = 1}^\infty 
{\mathcal{F}}{ {\mathcal{D}}^i}$ the Feynman-Dyson space over $J$ for $\mcH$ 
(the film for Feynman's spacetime events).  

\subsection{Time Ordered Operators}

If $\mcC({\mathcal{H}_ \otimes })$ is a set of closed densely defined linear 
operators on ${\mathcal{H}_ \otimes }$, and  $\{H_I(t), t \in J\}$ is a family 
of generators for unitary groups, we define  $\mcC(\mathcal{H}{ {(}}t{ {))}} 
\subset \mcC({\mathcal{H}_ \otimes })$ by:
\[\mcC{{(}}\mathcal{H}{{(t))  = }}\left\{ {{{\mathbf{H}}_I}({{t}})
\left| {{{ }}{{\mathbf{H}}_I}({{t}})} \right.} \right.\left. {{{ 
= }}\mathop {\overset{\lower0.5em\hbox{$\smash{\scriptscriptstyle\frown}$}}{ 
\otimes } }\limits_{T \geqslant s > t} {{\text{I}}_s} \otimes {H_I}(t) 
\otimes {{(}}\mathop  \otimes \limits_{t > s \geqslant -T} {{{I}}_s}{{)}}} \right\},\]
where $I_s$ is the identity operator at time $s$.  It follows that
\[
{{\mathbf{H}}_I}{ {(}}t{ {)}}{{\mathbf{H}}_I}(s) = {{\mathbf{H}}_I}{ 
{(}}s{ {)}}{{\mathbf{H}}_I}{ {(}}t{ {),  }}\;t \ne s.
\]
Thus, the operators are ordered in time, commute when acting at different times 
and maintain their mathematically defined positions.
 
\subsection{Time Ordered Integrals}

We now state the fundamental theorem for time-ordered operators (see \cite{GZ}):			
\begin{thm}\lb{fd1}{\tx{(Fundamental Theorem for Time-Ordered Operators)}} 
If $ \{ {H_I}(t){ { }}|{ { }}t \in J\}$ is a family of weakly continuous 
Hamiltonian generators of unitary groups on $\mcH$ and $\{ {\mathbf{H}}_I(t)| t 
\in J\}$ is the time ordered version defined on $\mcF{\mcD}_{\otimes}^2$ then: 
\begin{enumerate}
\item The family $\{ {\mathbf{H}}_I(t)| t \in J\}$ is strongly continuous and 
the time-ordered HK-integral ${\mathbf{Q}}\left[ {t, - T} \right] 
= \int_{ - T}^t {{{\mathbf{H}}_I}\left( s \right)ds}$ exists (a.e), has a 
dense domain and is the generator of a strongly continuous unitary 
group ${\bf{U}}[t, -T]$ on $\mcF{\mcD}_{\otimes}^2$. Moreover:
\item If $\Psi_0$ is in the domain of ${\mathbf{Q}}\left[ {t, - T} \right]$ 
then $\Psi \left( t \right) = {\mathbf{U}}\left[ {t, - T} \right]{\Psi _0} 
= \exp \left\{ { - \tfrac{i}{\hbar }{\mathbf{Q}}\left[ {t, - T} \right]} 
\right\}{\Psi _0}$ satisfies:
\begin{equation}
i{\hbar}\frac{{\partial \Psi (t)}}{{\partial t}} = {{\mathbf{H}}_I}
\left( t \right)\Psi \left( t \right),\;\Psi \left( { - T} \right) = {\Psi _0}. 
\label{(33)}
\end{equation}
\end{enumerate}
\end{thm} 

\subsection{Interaction Representation}

Haag's theorem shows that the (sharp) equal-time commutation relations of an 
interacting field are equivalent to those of a free field (Streater and Wightman 
\cite{SW}). Lindner et. al. have shown that there is experimental support 
for interference in time of the wave function for a particle (see Horwitz \cite{HO2}).   
In this section we show that if time smearing exists, the interaction representation is defined.
Before proving this result, we consider the concept of an exchange operator.
\begin{Def} An exchange operator ${{E}}[t,t']$, is defined for pairs $t,t'$ such that:
\begin{enumerate} 
\item
$E[t,t'] \,:\;\mcC[{\mathcal{H}}(t)] \to \mcC[{\mathcal{H}}( t')]$, (bijective mapping), 
\item $
E[s,t']E[t, s] = E[t,t'] $,
 \item
$
E[t,t'] E[t',t] = {\bf{I}}_\otimes,
$\, (identity)
\item  for $s \ne t,\;t'$,
$E[t,t'] {\bf{H}}_I(s) = {\bf{H}}_I(s)$, 
 for all ${\bf{H}}{_I}(s) \in \mcC[{\mathcal{H}}(s)].
$
\end{enumerate}
\end{Def}
Assume that ${\bm{H}}_0 (t) $ and ${\bm{H}}_1 (t) $ are generators of 
unitary groups for each $t \in J$, ${\bm{H}}_I(t) = {\bm{H}}_0 (t) \oplus 
{\bm{H}}_1 (t) $ is densely defined, ${\bm{H}}_1^n(t) = n{{\bm{H}}_1(t)}R\left( 
{n,{{\bm{H}}_1(t)}} \right)$ is the Yosida approximator for ${\bm{H}}_1 (t)$  
and Theorem (\ref{fd1}) is satisfied. Define ${\mathbf{U}}_n [t,a] $, ${\mathbf{U}}_0 
[t,a]$ and ${\mathbf{{U}}}_0^\ka  [t,a] $ as
\[
\begin{gathered}
  {\mathbf{U}}_n [t,a] = \exp \{ ( - {i \mathord{\left/
 {\vphantom {i \hbar }} \right.
 \kern-\nulldelimiterspace} \hbar })\int\limits_a^t {[ {\bm{H}}_0 (s)}  
\oplus {\bm{H}}_1^n (s)]ds\} , \hfill \\
  {\mathbf{U}}_0 [t,a] = \exp \{ ( - {i \mathord{\left/
 {\vphantom {i \hbar }} \right.
 \kern-\nulldelimiterspace} \hbar })\int\limits_a^t { {\bm{H}}_0 (s)} ds\} , \hfill \\
  {{\bm{U}}}_0^\ka [t,a] = \exp \{ ( - {i \mathord{\left/
 {\vphantom {i \hbar }} \right.
 \kern-\nulldelimiterspace} \hbar })\int\limits_a^t { {\bm{H}}_0^\ka (s)} ds\} , \hfill \\
 {\bm{H}}_0^\ka  (t) = \int\limits_{ - \infty }^\infty  {\rho _\ka  
(t,s){\mathbf{E}}[t,s]{{\bm{H}}}_0 (s)} ds, \hfill \\ 
\end{gathered} 
\]
where $\rho _\ka  (t,s) $ is the smearing density, which depends on the Planck length
$\ka$ and $\int_{ - \infty }^\infty  {\rho _\ka  (t,s)ds = 1} $ (for example, $
\rho _\ka  (t,s) = [{1 \mathord{\left/
 {\vphantom {1 {\sqrt {2i\pi \ka^2 } }}} \right.
 \kern-\nulldelimiterspace} {\sqrt {2i\pi \ka ^2 } }}]
\exp \{ {{i(t - s)^2 } \mathord{\left/
 {\vphantom {{i(t - s)^2 } {2\ka ^2 }}} \right.
 \kern-\nulldelimiterspace} {2\ka^2 }}\} 
$).   
 
We now obtain: 
\[
{\bm{H}}_{{I}}^n (t) = {\mathbf{{U}}}_0^\ka [a,t] {\bm{H}}_{\text{1}}^n (t)
{\mathbf{{U}}}_0^\ka  [t,a],
\]
and the terms do not commute. If we set 
$\Psi _n^{} (t) = {\mathbf{ U}}_0^\ka  [a,t]{\mathbf{U}}_n [t,a]\Phi$, 
we have
\[
\begin{gathered}
  \frac{\partial }
{{\partial t}}\Psi _n (t) = \frac{i}
{\hbar }{\mathbf{{U}}}_0^\ka  [a,t] {\bm{H}}_0 (t){\mathbf{U}}_n [t,a]\Phi  - \frac{i}
{\hbar }{\mathbf{{U}}}_0^\ka  [a,t]\left[ { {\bm{H}}_0 (t) 
+ {\bm{H}}_1^n (t)} \right]{\mathbf{U}}_n [t,a]\Phi  \hfill \\
 {\text{so that}}\quad  \frac{\partial }
{{\partial t}}\Psi _n (t) = \frac{i}
{\hbar }\{ {\mathbf{{U}}}_0^\ka  [a,t] {\bm{H}}_1^n (t){\mathbf{{U}}}_0^\ka  
[t,a]\} {\mathbf{{U}}}_0^\ka  [a,t]{\mathbf{U}}_n [t,a]\Phi  \hfill \\
{\text{and}}\quad  i\hbar \frac{\partial }
{{\partial t}}\Psi _n (t) = {\bm{H}}_{{I}}^n (t)\Psi _n (t),\;\Psi _n (a) = \Phi . \hfill \\ 
\end{gathered} 
\]
With the same conditions as Theorem (\ref{fd1}), we have
\begin{thm} If $Q_1 [t,a] = \int_a^t {H_1 (s)ds}$  generates a unitary group 
on ${\mathcal{H}}$, then the time-ordered integral 
${\mathbf{Q}}_{{I}} [t,a] = \int_a^t { {\bm{H}}_{{I}} (s)ds}$, where 
${\bm{H}}_{{I}} (t) = {\mathbf{{U}}}_0^\ka  [a,t] {\bm{H}}_1 (t){\mathbf{{U}}}_0^\ka  
[t,a]$ generates a unitary group on ${\mathcal{F}\mathcal{D}}_ \otimes ^2 $, and 
\[
\exp \{ ( - {i \mathord{\left/
 {\vphantom {i \hbar }} \right.
 \kern-\nulldelimiterspace} \hbar }){\mathbf{Q}}_{{I}}^n [t,a]\}  
\to \exp \{ ( - {i \mathord{\left/
 {\vphantom {i \hbar }} \right.
 \kern-\nulldelimiterspace} \hbar }){\mathbf{Q}}_{{I}} [t,a]\}, 
\]
where 
${\mathbf{Q}}_{{I}}^n [t,a] = \int_a^t { {\bm{H}}_{{I}}^n (s)ds}$, and:
\[
i\hbar \frac{\partial }
{{\partial t}}\Psi (t) = {\bm{H}}_{{I}} (t)\Psi (t),\;\Psi (a) = \Phi. 
\]
\end{thm}

\subsection{{Feynman Path Integral I}}{\lb{fI}}

Recall that ${K}{S}^2 [{{\R}}_I^n ]= {K}{S}^2[{{\R}}_I^n ]^*\supset 
{{L}}^1[{{\R}}_I^n ]^{**} = \mathfrak{M}[{{\R}}_I^n ]$, the space of finitely  
(and countable) additive measures on ${{\R}}_I^n$. (It also contains 
$\mcD^*[{{\R}}_I^n ]$, the space of distributions.)  Thus, the Schr\"{o}dinger 
equation with Dirac measure initial data is well posed on ${K}{S}^2 [{{\R}}_I^n ]$.  
In this section we construct the Feynman path integral in a manner that preserves 
both intuitive and computational advantages. We begin with a few 
operator extensions to ${KS}^2[{{\R}}_I^n ]$.

\subsubsection{\bf{Extensions to ${K}{S}^2 [{{\R}}_I^n ]$}}

Bearing in mind that the position operator ${\bf{x}}$ and momentum operator 
${\bf{p}}$ are closed and densely defined on ${{L}}^2 [{{\R}}^n ]$ and, 
our extension operator $T_e$ extends this property to ${{L}}^2 [{{\R}}_I^n ]$, 
it is easy to see that both have closed, densely defined extensions to 
${K}{S}^2 [{{\R}}_I^n ]$.  If $f, g \in L^1[\R_I^n]$, we denote the 
Fourier transform of $f$ and the convolution of $g$ with respect to $f$ 
by $\mfF(f)$ and $\mfC_f(g)$, respectively. The following theorem extends 
them to bounded linear operators on ${K}{S}^2 [{{\R}}_I^n ]$, ensuring that both 
the Schr\"{o}dinger and Heisenberg theories have faithful representations on 
${K}{S}^2 [{{\R}}_I^n ]$.  
\begin{thm}\lb{extend} For each $f, g \in L^1[\R_I^n] \vee L^2[\R_I^n]$, the Fourier 
transform $\mfF(f)$ and the convolution of $g$ with respect to $f$, $\mfC_f(g)$, 
both extend to ${KS}^2 [{\R}_I^n]$ as bounded linear operators.
\end{thm}
If $B \in \mcB(\R_I^n)$ is bounded, ${\bf x} \in \R_I^n$ fixed and 
$t, s \in \R$, then the kernel:
\[
\mathbb{K}_{\mathbf{f}} [t,{\mathbf{x}}\,;\,s,B] = \int_B {\left( {2\pi i(t - s)} 
\right)^{ - 1/2} \exp \{ i{{\left| {{\mathbf{x}} - {\mathbf{y}}} \right|^2 } \mathord{\left/
{\vphantom {{\left| {{\mathbf{x}} - {\mathbf{y}}} \right|^2 } {2(t - s)}}} \right.
 \kern-\nulldelimiterspace} {2(t - s)}}\} d{\mathbf{y}}} 
\]
is in ${K}{S}^2 [{{\R}}_I^n ] $, it is a finitely additive measure with
$\left\| {\mathbb{K}_{\mathbf{f}} [t,{\mathbf{x}}\,;\,s,B]} 
\right\|_{{{KS}^2}}  \leqslant 1$ and  
\[
\mathbb{K}_{\mathbf{f}} [t,{\mathbf{x}}\,;\,s,B] = \int_{{{\R}}_I^n } 
{\mathbb{K}_{\mathbf{f}} [t,{\mathbf{x}}\,;\,\tau ,d{\mathbf{z}}]
\mathbb{K}_{\mathbf{f}} [\tau ,{\mathbf{z}}\,;\,s,B]}, \: \:{\text {(HK-integral)}}. 
\]
\begin{Def}Let ${\mathbf{P}}_n  = \{ t_0 ,\tau _1 ,t_1 ,\tau _2 , \cdots ,
\tau _n ,t_n \}$ one of a family of HK-$\de_n$ partitions of the interval $[0,t] $ for each $n$, 
with $\limsup _{n \to \infty } \delta_n(\tau)  = 0$.  Set $\Delta t_j  
= t_j  - t_{j - 1}, \tau _0  = 0$ and for $
\psi  \in {K}{S}^2 [{{\R}}_I^n ] $ define
\begin{equation}
\int_{{{\R}}_I^{n[0,t]} } {\mathbb{K}_{\mathbf{f}} [{\mathcal{D}}_\lambda  
{\mathbf{x}}{\text{(}}\tau ){\text{ ; }}{\mathbf{x}}{\text{(}}0{\text{)}}]}  
= e^{ - \lambda t} \sum\limits_{k = 0}^{ [|{\lambda t}|] } {\frac{{\left( 
{\lambda t} \right)^k }}
{{k!}}\left\{ {\prod\limits_{j = 1}^k {\int_{{{\R}}_I^n } {\mathbb{K}_
{\mathbf{f}} [t_j ,{\mathbf{x}}{\text{(}}\tau _j {\text{)}}\,;\,t_{j - 1} 
,d{\mathbf{x}}{\text{(}}\tau _{j - 1} {\text{)}}]} } } \right\}}, 
\label{(34)}
\end{equation}
and 
\begin{equation}
 \begin{gathered}
  \int_{{{\R}}_I^{n[0,t]} } {\mathbb{K}_{\mathbf{f}} [\mathcal{D}  
{\mathbf{x}}(\tau );{\mathbf{x}}(0)]\psi [{\mathbf{x}}(0)]}  \hfill \\
= \mathop {\lim }\limits_{\lambda  \to \infty } \int_{{{\R}}_I^{n[0,t]} } 
{\mathbb{K}_{\mathbf{f}} [\mathcal{D}_\lambda  {\mathbf{x}}(\tau );
{\mathbf{x}}(0)]\psi [{\mathbf{x}}(0)]}  \hfill \\ 
\end{gathered} 
\label{(35)}
\end{equation}
whenever the limit exists.
\end{Def}
The following is clear. A general result is presented in the next section. 
\begin{thm}The function $\psi ({\mathbf{x}}) \equiv 1 \in {K}{S}^{2} [{{\R}}_I^n ]$ and 
\begin{equation}
\int_{{{\R}}_I^{n[s,t]} } {\mathbb{K}_{\mathbf{f}} [{\mathcal{D}}
{\mathbf{x}}{\text{(}}\tau ){\text{ ; }}{\mathbf{x}}{\text{(}}s{\text{)}}]}  
= \mathbb{K}_{\mathbf{f}} [t,{\mathbf{x}}\,;\,s,{\mathbf{y}}] = \tfrac{1}
{{\sqrt {2\pi i(t - s)} }}\exp \{ i{{\left| {{\mathbf{x}} - {\mathbf{y}}} \right|^2 } \mathord{\left/
{\vphantom {{\left| {{\mathbf{x}} - {\mathbf{y}}} \right|^2 } {2(t - s)}}} \right.
\kern-\nulldelimiterspace} {2(t - s)}}\}. 
\label{(36)}
\end{equation}
\end{thm}
The above result is exactly what Feynman expected.

\subsection{Path Integral II} 

Feynman suggested that his path integral is a special case of the time-ordered 
operator calculus. In this section we consider the advantages of time ordering.
\begin{thm} 
If $U[t,a]$ is any well defined evolution operator on ${{KS}}^2({\R}_I^n)$, with a time-dependent 
generator $H(t) $, and reproducing kernel $ {{K}}[{\mathbf{x}}(t), t\, 
; \, {\mathbf{x}}(s), s] $ such that:  
\[
\begin{gathered}
{{K}}\left[ {{\mathbf{x}}{\text{(}}t{\text{), }}t{\text{; }}{\mathbf{x}}
{\text{(}}s{\text{), }}s} \right] = \int_{{{\R}}_I^n }^{} {{{K}}\left[ 
{{\mathbf{x}}{\text{(}}t{\text{), }}t{\text{; }}d{\mathbf{x}}(\tau )
{\text{, }}\tau } \right]} {{K}}\left[ {{\mathbf{x}}{\text{(}}\tau {\text{), }}
\tau {\text{; }}{\mathbf{x}}{\text{(}}s{\text{), }}s} \right], \hfill \\
U[t,a]\varphi (a) = \int_{{{\R}}_I^n }^{} {{{K}}\left[ {{\mathbf{x}}
{\text{(}}t{\text{), }}t{\text{; }}d{\mathbf{x}}(s){\text{, }}s} \right]\varphi (s)}, \;{\tx{such that}} \hfill \\ 
\frac{\partial }{\partial t} U\left[ t, a\right]  \varphi \left( a\right)  
=\dot{u} \left( t\right)  =H\left( t\right)  u\left( t\right)  ,\  
u\left( a\right)  =\varphi \left( a\right)  .    \hfill \\
\end{gathered} 
\]
Then the corresponding  time-ordered version ${\mathbf{U}}{{[}}t{{,}}s{{]}}$  
 defined on ${\mathcal{F}\mathcal{D}}_ \otimes ^{2}  
\subset {\mathcal{H}}_ \otimes ^{2}$, with kernel $\mathbb{K}_{\mathbf{f}} 
\left[ {{\mathbf{x}}{\text{(}}t{\text{), }}t{\text{; }}{\mathbf{x}}
{\text{(}}s{\text{), }}s} \right] $ satisfies the conditions of our fundamental theorem. 
\end{thm}
Since ${\mathbf{U}}{\text{[}}t
{\text{,}}\tau {\text{]}}{\mathbf{U}}{\text{[}}\tau {\text{,}}s{\text{] 
= }}{\mathbf{U}}{\text{[}}t{\text{,}}s{\text{]}}$, 

we have:
\[
\mathbb{K}_{\mathbf{f}} \left[ {{\mathbf{x}}{\text{(}}t{\text{), }}
t{\text{; }}{\mathbf{x}}{\text{(}}s{\text{), }}s} \right] 
= \int_{{{\R}}_I^n }^{} {\mathbb{K}_{\mathbf{f}} \left[ {{\mathbf{x}}
{\text{(}}t{\text{), }}t{\text{; }}d{\mathbf{x}}(\tau ){\text{, }}\tau } 
\right]} \mathbb{K}_{\mathbf{f}} \left[ {{\mathbf{x}}{\text{(}}\tau {\text{), 
}}\tau {\text{; }}{\mathbf{x}}{\text{(}}s{\text{), }}s} \right].
\]
From our sum over paths representation for ${\mathbf{U}}[t,s]$, we have:
\[
\begin{gathered}
{\mathbf{U}}[t,s]\Phi (s) = \lim _{\lambda  \to \infty } 
{\mathbf{U}}_\lambda  [t,s]\Phi (s) \hfill \\
  {\text{        }} = {\text{lim}}_{\lambda  \to \infty } 
\operatorname{e} ^{ - \lambda \left( {t - s} \right)} \sum\limits_{k = 0}^n  
{\frac{{\mathop {\left[ {\lambda \left( {t - s} \right)} \right]}\nolimits^k }}
{{k!}}} {\mathbf{U}}_k [t,s]\Phi (s), \hfill \\ 
\end{gathered} 
\]
where $n = [|\lambda (t - s)|]$ (the greatest integer in $\lambda (t - s) $) and
\[
{\mathbf{U}}_k [t,s]\Phi (s) = \exp \left\{ {( - i/\hbar )\sum\limits_{j 
= 1}^k {\int_{t_{j - 1} }^{t_j } {\bm{E}}[\tau_j, \tau] 
{\bf{H}}(\tau )d\tau } } \right\}\Phi (s).
\]
As before, define $\mathbb{K}_{\mathbf{f}} [{\mathcal{D}}_\lambda 
{\mathbf{x}}{\text{(}}t ){\text{ ; }}{\mathbf{x}}{\text{(s)}}]$ by
\[
\begin{gathered}
 {\mathbb{K}_{\mathbf{f}} [{\mathcal{D}}_\lambda  {\mathbf{x}}
{\text{(}}t ){\text{ ; }}{\mathbf{x}}{\text{(s)}}]}  \hfill \\
  {\text{         }} = : {e^{ - \lambda \left( {t - s} \right)}}
\sum\limits_{k = 1}^n {\tfrac{{\left[ {\lambda \left( {t - s} \right)} 
\right]}}{{k!}}\left\{ {\prod\limits_{j = 1}^k {\int_{{\mathbb{R}_I^n}} {{\mathbb{K}_f}
\left[ {{t_j},{\mathbf{x}}\left( {{t_j}} \right);d{\mathbf{x}}\left( {{t_{j - 1}}} 
\right),{t_{j - 1}}} \right]\left| {^{{\tau _j}}} \right.} } } \right\}} \hfill \\ 
\end{gathered} 
\]	
and $|^{\tau_j}$ denotes that the integration is performed at the time $\tau_j$.
\begin{Def} We define the Feynman path integral associated with ${\mathbf{U}}[t,s]$ by:
\[
{\mathbf{U}}[t,s]\Phi (s) = \int_{{{\R}}_I^{n[t,s]} } {\mathbb{K}_{\mathbf{f}} 
[{\mathcal{D}}{\mathbf{x}}{\text{(}}\tau ){\text{ ; }}{\mathbf{x}}{\text{(s)}}]}\Phi (s)  
= \lim _{\lambda  \to \infty } \int_{{{\R}}_I^{n[t,s]} } {\mathbb{K}_{\mathbf{f}} 
[{\mathcal{D}}_\lambda  {\mathbf{x}}{\text{(}}\tau ){\text{ 
; }}{\mathbf{x}}{\text{(s)}}]}\Phi (s). 
\]
\end{Def}
\begin{thm} For the Feynman time-ordered theory, whenever a reproducing kernel 
exists on $KS^2[\R_I^n]$, we have
\[
\mathop {\lim }\limits_{\lambda  \to \infty } {\mathbf{U}}_\lambda  
[t,s]\Phi (s) = {\mathbf{U}}[t,s]\Phi (s) = \int_{{{\R}}_I^{n[t,s]} } 
{\mathbb{K}_{\mathbf{f}} [{\mathcal{D}}^\lambda  {\mathbf{x}}{\text{(}}
\tau ){\text{ ; }}{\mathbf{x}}{\text{(s)}}]\Phi [{\mathbf{x}}{\text{(}}s{\text{)}}]} ,
\]
and the limit is independent of the space of continuous functions that we choose.
\begin{rem}  This result includes the Wiener path integral, which requires 
additional effort to restrict it to the space of continuous paths.  We also 
note that the base space $\mcH$ is allowed to continuously change in time 
and the family of spaces need not be all Hilbert. In addition, the intersection 
of the corresponding domains of the generators can even be 
the empty set (see Goldstein \cite{GS}).
\end{rem}
\end{thm}\paragraph{}Let us assume that $H_0 (t)$ and $H_1 (t)$ are strongly 
continuous generators of unitary groups, with a common dense domain $D(t)$, 
for each $t \in J = [a,b]$, where 
${\bm{H}}_{1,\rho } (t) = \rho {\bm{H}}_1 (t){{R}}(\rho , {\bm{H}}_1 (t))$ 
is the Yosida approximator for the time-ordered version of 
$H_1 (t)$, with dense domain $D= \otimes_{t \in I}{D(t)}$.  
Define ${\mathbf{U}}^\rho  [t,a] $ and ${\mathbf{U}}^0 [t,a]$ as follows: 
\[
\begin{gathered}
  {\mathbf{U}}^\rho  [t,a] = \exp \{ ( - {i \mathord{\left/
 {\vphantom {i \hbar }} \right.
 \kern-\nulldelimiterspace} \hbar })\int\limits_a^t {[ {\bm{H}}_0 (s)}  
+ {\bm{H}}_{1,\rho } (s)]ds\} , \hfill \\
  {\mathbf{U}}^0 [t,a] = \exp \{ ( - {i \mathord{\left/
 {\vphantom {i \hbar }} \right.
 \kern-\nulldelimiterspace} \hbar })\int\limits_a^t { {\bm{H}}_0 (s)} ds\} . \hfill \\ 
\end{gathered} 
\]
Because ${\bm{H}}_{1,\rho } (s) $ is bounded, ${\bm{H}}_0 (s) + {\bm{H}}_{1,\rho } (s)$ 
is a generator of a unitary group for each $s \in J$ and a finite $\rho $.  
Now assume that ${\mathbf{U}}^0 [t,a] $ has an associated reproducing kernel such that
${\mathbf{U}}^{0} [t,a] = \int_{{{\R}}_I^{n[t,s]}}{\mathbb{K}}_{\mathbf{f}} 
[{\mathcal{D}}{\mathbf{x}}(\tau) ; {\mathbf{x}}(a)]$. 
We now have the following general result.  
\begin{thm}{\lb{pII}} (Extended Feynman-Kac) If ${\bm{H}}_0 (s) \oplus {\bm{H}}_1 (s)$ 
is a generator of a unitary group, then
\[
\begin{gathered}
\mathop {\lim }\limits_{\rho  \to \infty } {\mathbf{U}}^\rho  [t,a]\Phi (a) 
= {\mathbf{U}}[t,a]\Phi (a) \hfill \\
= \int_{{{\R}}_I^{n[t,a]} } {\mathbb{K}_{\mathbf{f}} [{\mathcal{D}}{\mathbf{x}}
{\text{(}}\tau ){\text{ ; }}{\mathbf{x}}{\text{(a)}}]\exp \{ ( - {i \mathord{\left/
 {\vphantom {i \hbar }} \right.
 \kern-\nulldelimiterspace} \hbar })\int\limits_a^\tau  { {\bm{H}}_1 (s)ds} ]\} 
\Phi [{\mathbf{x}}{\text{(}}a{\text{)}}]} . \hfill \\ 
\end{gathered} 
\]
\end{thm}
\begin{proof} The fact that ${\mathbf{U}}^\rho  [t,a]\Phi (a) \to 
{\mathbf{U}}[t,a]\Phi (a) $ is clear. To prove that 
\[
{\mathbf{U}}[t,a]\Phi (a) = \int_{{{\R}}_I^{n[t,a]} } {\mathbb{K}_{\mathbf{f}} 
[\mathcal{D}{\mathbf{x}}(\tau );{\mathbf{x}}(a)]} \exp \{ ( - i/\hbar )
\int_a^t {\bm{H}_1 (s)ds} \}\Phi (a).
\]
First, note that because the time-ordered integral exists and we are only interested 
in the limit, we can write for each $k$,
\[
U_k^\rho  [t,a]\Phi (a) = \exp \left\{ {( - i/\hbar )\sum\nolimits_{j = 1}^k 
{\int_{t_{j - 1} }^{t_j } {\left[ {{\mathbf{E}}[\tau _j ,s]\bm{H}_0 (s) 
+ {\mathbf{E}}[\tau '_j ,s]\bm{H}_{1,\rho } (s)} \right]ds} } } \right\}\Phi (a),
\]
where $\tau_j $ and $\tau'_j $ are distinct points in $(t_{j-1}, t_j)$.  
Thus, we can also write $U_k^\rho  [t,a]\Phi (a)$ as 
\[
\begin{gathered}
{\mathbf{U}}_k^\rho [t,a] = exp\left\{ {\tfrac{{ - i}}{\hbar }\sum\limits_{j = 1}^k 
{\int_{{t_{j - 1}}}^{{t_j}} {{\mathbf{E}}[{\tau _j},s]{{\bm{H}}_0}(s)ds} } } 
\right\}exp\left\{ {\tfrac{{ - i}}{\hbar }\sum\limits_{j = 1}^k {\int_{{t_{j 
- 1}}}^{t_j} {{\mathbf{E}}[\tau'_j,s]{{\bm{H}}_{1,\rho }}(s)ds} } } \right\} \hfill \\
   = \prod\limits_{j = 1}^k {exp\left\{ {\tfrac{{ - i}}{\hbar }\int_{{t_{j - 1}}}^{{t_j}} 
{{\mathbf{E}}[{\tau _j},s]{{\bm{H}}_0}(s)ds} } \right\}exp\left\{ {\tfrac{{ - i}}{\hbar }
\sum\limits_{j = 1}^k {\int_{{t_{j - 1}}}^{tj} {{\mathbf{E}}[\tau' _j,s]
{{\bm{H}}_{1,\rho }}(s)ds} } } \right\}}  \hfill \\
   = \prod\limits_{j = 1}^k {\int_{{\mathbb{R}_I^n}} {{\mathbb{K}_{\mathbf{f}}}
\left[ {{t_j},{\mathbf{x}}({t_j})\;;\;{t_{j - 1}},d{\mathbf{x}}({t_{j - 1}})} 
\right]} } exp\left\{ {\tfrac{{ - i}}{\hbar }\sum\limits_{j = 1}^k 
{\int_{{t_{j - 1}}}^{t_j} {{\mathbf{E}}[\tau' _j,s]{{\bm{H}}_{1,\rho }}(s)ds} 
} } \right\}. \hfill \\ 
\end{gathered} 
\]
If we include this in our candidate evolution operator ${\mathbf{U}}_{\lambda}^{\rho} 
[t,a] \Phi (a)$ and compute the limit, we have
\[
\begin{gathered}
{\mathbf{U}}^\rho  [t,a]\Phi (a) \hfill \\
= \int_{{{\R}}_I^{n[t,a]} } {\mathbb{K}_{\mathbf{f}} [\mathcal{D}{\mathbf{x}}(t);
{\mathbf{x}}(a)]} \exp \left\{ {( - i/\hbar )\int_a^t {\bm{H}_{1,\rho } (s)ds} } 
\right\}\Phi (a). \hfill \\ 
\end{gathered} 
\]
Since the limit as $\rho  \to \infty$ exists on the left, it defines the limit on the right.
\end {proof}

\subsection{Examples}  
Theorem {\rf{pII}} is somewhat abstract. The following example covers most of  
non-relativistic quantum theory.  
\begin{thm}
Let ${\bm\Delta}$ be the Laplacian on $L^2[{{\R}}_I^n]$, or some other Hilbert space, 
$\mcH$ and let $V$ be any potential such that $H=(-\hbar^2/2){\bm\Delta} \oplus V$ 
generates a unitary group on $\mcH$ (see remarks below).
Using time as an index, the problem
\[
(i\hbar) {\partial\psi({\bf{x}},t)}/{\partial t}= {\bm H}(t)
\psi({\bf{x}},t),\;\;  \psi({\bf{x}},0)=\psi_{0}({\bf{x}}),
\]
has a a unique solution on $\mcF{\mcD}_{\otimes}^2$ with the extended Feynman-Kac representation.
\end{thm} 
\begin{rem}We have used $\oplus$ to allow a generalized definition of addition 
(i.e, Trotter-Kato). In fact, Kato has shown that $V$ can be any self-adjoint 
generator and Goldstein has called it a generalized Lie sum (see \cite{KA}) 
\end{rem}
Our second example is due to Albeverio and Mazzucchi \cite{AM}. It 
is like the first but we provide a different approach.   
Let $\mathbb{C}$ be a completely symmetric positive definite fourth-order 
covariant tensor on $L^2[{{\R}}_I^n]$, let $\Omega$ be a symmetric positive-definite 
$n \times n$ matrix and let $\lambda $ be a nonnegative constant. Then:  
\[
 H =  - \tfrac{\hbar^2 }{2}{\bm\Delta}  + \tfrac{1}
{2}{\mathbf{x}}\Omega ^2 {\mathbf{x}} + \lambda \mathbb{C}[\mathbf{x},
\mathbf{x},\mathbf{x},\mathbf{x}]
\]
is known to be the generator for a unitary group on $L^2 [{{\R}}_I^n ]$.  
Albeverio and Mazzucchi \cite{AM} prove that $\bar H$ (closure) has a path integral 
representation as the analytic continuation (in the parameter $\lambda $) 
of an infinite dimensional generalized oscillatory integral.  
(Their version of Feynman's path integral.)  

Using the results of the previous sections, we can extend $H$ to 
${K}{S}^2 [{{\R}}_I^n ]$, which generates a unitary group. Let 
$V = \tfrac{1}{2}{\mathbf{x}}\Omega ^2 {\mathbf{x}} + \lambda 
\mathbb{C}[{\mathbf{x}},{\mathbf{x}},{\mathbf{x}},{\mathbf{x}}]$ and $V_\rho   
= V(I + \rho V^*V )^{ - 1/2} ,{\text{ }}\rho  > 0$. We can prove that $V_\rho$ is 
a bounded generator that converges to $V$. Since $ - \tfrac{\hbar^2 }{2}\bm\Delta$ 
generates a unitary group,  $H_\rho   =  - \tfrac{\hbar^2 }{2}{\bm\Delta}  
+ V_\rho$ also generates one and converges to $H$.  Let
\[ 
{\bm{H}}(\tau) = (\mathop {\hat  \otimes }\limits_{t \geqslant s > \tau} 
{\text{I}}_s)  \otimes H \otimes (\mathop  \otimes \limits_{\tau > 
s \geqslant 0} {\text{I}}_s ),
\]
then ${\bm{H}}(t)$ generates a unitary group for each $t$ and ${\bm{H}}_\rho(t)$ 
converges to ${\bm{H}}(t)$ on ${\mathcal{F}\mathcal{D}}_ \otimes ^{2} $.  
We can now obtain:
\[
\begin{gathered}
{\mathbf{U}}[t,a] \Phi= \int_{{{\R}}_I^{n[t,a]} } {\mathbb{K}_{\mathbf{F}} 
[{\mathcal{D}}{\mathbf{x}}(\tau )\,;\,{\mathbf{x}}(a)]} \exp \lt\{  - ({i \mathord{\left/
 {\vphantom {i \hbar }} \right.
 \kern-\nulldelimiterspace} \hbar })\int_a^\tau  {V(s)ds} \rt\} \Phi  \hfill \\
= \lim _{\rho  \to 0} \int_{{{\R}}_I^{n[t,a]} } {\mathbb{K}_{\mathbf{F}} 
[{\mathcal{D}}{\mathbf{x}}(\tau )\,;\,{\mathbf{x}}(a)]} \exp \lt\{  - ({i \mathord{\left/
 {\vphantom {i \hbar }} \right.
 \kern-\nulldelimiterspace} \hbar })\int_a^\tau  {V_\rho  (s)ds} \rt\} \Phi . \hfill \\ 
\end{gathered} 
\]
We refer to \cite{GZ} for additional examples, including path integrals for kernels 
that are not perturbations of the Laplacian.

\section{Concluding remarks}

Modern physical theories of fundamental interactions need the space
$\Phi$ of all field histories over spacetime \cite{DW4}. 
This is an infinite-dimensional manifold, and 
the framework of field theory makes it therefore compelling
to look for an appropriate mathematical language. This can be
obtained by the choice of a separable Banach space, following
the beautiful and profound presentation of Geroch \cite{GER}, with a
subsequent application of the results presented in the main body of
our paper. The infinite-dimensional setting is still an unknown
land for the majority of the physics community. For example,
contraction of tensor arguments is not defined therein \cite{GER},
and hence all geometric invariants which contribute to ultraviolet
divergences of gravity \cite{THO,GOR} cannot be defined. 

\subsection{Manifolds modelled on Banach spaces}

Within this framework, the idea is to pick out, from the rich
structure of a Banach space, a particular type of structure called
by Geroch {\it local smoothness} \cite{GER}. As a first step, one has to
introduce a mechanism by means of which structure can be carried
from Banach spaces to other mathematical entities. If $M$ is a set
and $E$ is a Banach space, an $E$-chart on $M$ consists of a subset
$U$ of $M$ jointly with a map $\psi$ from $U$ to $E$, such that
$\psi$ is one-to-one, and the image $\psi(U)$ of $U$ by $\psi$ is
open in $E$. A chart establishes therefore a one-to-one correspondence
between a certain subset $U$ of $M$ and a certain open subset
$\psi(U)$ of $E$. It is this correspondence that carries structure
from $E$ to $M$. 

As a second step, one needs the concept of {\it agreement
between two charts as regards their induced smoothness structures
on $M$}. Let $(U,\psi)$ and $(U',\psi')$ be two $E$-charts on
the set $M$. On the intersection $V = U \cap U'$, two smoothness
structures are induced, one from $\psi$ and the other from $\psi'$.
The former defines a correspondence between $V$ and the subset
$\psi(V)$ of $E$, while the latter defines a correspondence
between $V$ and the subset $\psi'(V)$ of $E$. 

In order to compare these smoothness structures, let us consider
the map $\psi' \odot \psi^{-1}$ from $\psi(V)$ to $\psi'(V)$, and
its inverse $\psi \odot \psi'^{-1}$, from $\psi'(V)$ to $\psi(V)$.
These maps describe the interaction between the $E$-charts
$(U,\psi)$ and $(U',\psi')$. {\it At this stage, we have got rid of the
manifold $M$, and we deal with maps between subsets of Banach spaces}.
Now a mathematical symbol $p$ is fixed, either a non-negative 
integer or the $\infty$ symbol, and the charts $(U,\psi)$ and
$(U',\psi')$ on $M$ are said to be $C^{p}$-compatible if the images
$\psi(V)$ and $\psi'(V)$ are both open subsets of $E$, and if the maps
$$
\psi' \odot \psi^{-1}: \; \psi(V) \rightarrow E, \; \;
\psi \odot {\psi'}^{-1}: \; \psi'(V) \rightarrow E
$$
are both $C^p$ maps. The key role is played by the second condition.
Instead of requiring that our maps preserve vector space or norm
structure, we just require that they preserve $C^p$ differential
structure. In such a way, a single type of structure is isolated.

A manifold modelled on a Banach space consists of a non-empty set
$M$, a Banach space $E$, a symbol $p$, and a collection $\zeta$ of
$E$-charts on $M$, in such a way that the following conditions
are satisfied:
\vskip 0.3cm
\noindent
(i) Any two charts in $\zeta$ are $C^p$ compatible.
\vskip 0.3cm
\noindent
(ii) The charts in $\zeta$ cover the set $M$, i.e., every point
of $M$ lies in at least one of the $U$'s-
\vskip 0.3cm
\noindent
(iii) Any chart on $M$ which is compatible with all the charts
in $\zeta$ is itself an element of $\zeta$.
\vskip 0.3cm
\noindent
(iv) Given distinct points $p$ and $p'$ of $M$, there exist
charts $(U,\psi)$ and $(U',\psi')$ in $\zeta$ such that $p$
lies in $U$ and $p'$ lies in $U'$, and such that there exists
a ball $B$ centred at $\psi(p)$ in $\psi(U)$ and a ball $B'$
centred at $\psi'(p')$ in $\psi'(U')$, with the inverse images
$\psi^{-1}(B)$ and ${\psi'}^{-1}(B')$ having empty intersection
in $M$. The charts are then said to separate points of $M$.

Condition (i) means that, whenever two charts in the collection
$\zeta$ induce smoothness structures in the same region of $M$,
these structures agree. Condition (ii) requires that smoothness
structure has been induced over all of $M$. Condition (iii)
makes sure that no additional structure has been induced on $M$.
Last, condition (iv) rules out non-Hausdorff manifolds.
These four conditions define the concept of $C^{\infty}$ manifold
$M$ based on a Banach space $E$. The charts in $\zeta$ are said to
be the admissible charts. A subset $O$ of the manifold $M$ is
said to be open if, for every admissible chart $(U,\psi)$, the
image $\psi(U \cap O)$ is open in $E$. These are just the open
sets for a topology on $M$.

\subsection{Further perspectives}

In the physics-oriented literature, the object of interest is the
in-out amplitude with its functional integral formal representation
$$
\langle {\rm out}|{\rm in} \rangle= \int 
e^{i ({\dot S}[\varphi]+{1 \over 2}\omega_{\alpha \beta}
K^{\alpha}[\varphi] K^{\beta}[\varphi])}
({\rm det} {\widehat G}(\varphi))^{-1}
{\dot \mu}[\varphi][d\varphi],
$$
where $\dot S$ denotes the classical action $S$ supplemented by
all counterterms that are needed to render the amplitude 
finite \cite{DW4}. It is the factor
$({\rm det} {\widehat G}[\varphi])^{-1}$ that gives rise to
all ghost loops in the loop expansion, and the exponent $-1$
is what makes the ghost field a fermion. Such a ghost arises 
entirely from the fibre-bundle structure of the space
$\Phi$ of field histories, from the Jacobian of the
transformation from fibre-adapted coordinates to local
fields $\varphi^{i}$ \cite{DW4}. 

If the action functional for the in-out amplitude can be 
expanded to quadratic order in gauge and ghost fields, it
can be given a rigorous meaning as a Henstock integral, by
virtue of the material presented from Sec. 2 to Sec. 6.
However, as far as we can see, a functional integral for
nonperturbative gauge theories is still beyond the present
capabilities of the framework presented in our paper.

If the choice of infinite-dimensional manifolds modeled
on separable Banach spaces is the right thing to do,
we expect that this will be the starting point for further
progress on functional integrals for quantum gauge theories
(see also the remarkable monograph by Glimm and Jaffe \cite{GL}).

For further literature on the framework that we have outlined,
we refer the readers to the work in all References not cited so far.
\vskip 0.3cm
\leftline {\bf Acknowledgment}
\vskip 0.3cm
\noindent
Giampiero Esposito is grateful to INDAM for membership.

\section{Appendix}

At the risk of some repetitions,
this appendix is aimed at physics-oriented readers who would appreciate
a self-contained description of concepts of real and functional analysis.

\subsection{Completion of a normed space}

For any normed vector space $S$ it is always possible to build a
Banach space $S_1$ in such a way that $S$ contains a linear space
$\Sigma$ dense upon $S_1$ and equivalent to $S$. If $S$ is not
complete, every space $S_1$ as above is said to be completion of $S$.
In order to prove the theorem \cite{MIR}, let $E$ be the space having as elements
the sequences $\{ f_k \}$ of elements of $S$ which fulfill the Cauchy
condition, and let us introduce in $E$ an equivalence relation 
$\mathcal R$ upon assuming by definition that $\{ f_k \}$ and
$\{ f_{k}' \}$ are equivalent if 
\begin{equation}
\lim_{k \to \infty} \| f_k - f_{k}' \| =0.
\label{(37)}
\end{equation}
Let us set $S_{1}=E / {\mathcal R}$, with the letter $X$ used to denote its
elements. If $\{ f_k \}$ and $\{ f_{k}' \}$ are two representatives of the
element $X$ of $S_1$, the sequences $ \{ cf_{k} \}$ and 
$\{ c f_{k}' \}$ are, for all $c \in K$, equivalent to each other and
hence represent the same element of $S_1$ that we denote by $cX$. 
Thus, if we consider another element $Y$ of $S_1$, upon denoting by
$h_k$ and $h_{k}'$ two of its representatives, also the sequences
$\{ f_k + h_k \}$ and $\{ f_{k}' + h_{k}' \}$ are equivalent to each
other and hence represent the same element of $S_1$, that we denote
by $X+Y$.
 
By virtue of such definitions $S_1$ is a vector space. In particular,
the origin of $S_1$ is represented by any sequence $\{ \omega_k \}$
such that 
$$
\lim_{k \to \infty} \| \omega_{k} \| =0.
$$
Let us now show that in $S_1$ one can define a norm by setting 
\begin{equation}
\| X \|_{S_{1}} =\lim_{k \to \infty} \| f_{k} \|_{S},
\label{(38)}
\end{equation}
where $\{ f_{k} \}$ is an arbitrary representative of $X$. Since
$\{ f_{k} \}$ satisfies the Cauchy condition, and by virtue
of the inequality 
\begin{equation}
\left | \| f_{k+p} \|_{S} - \| f_{k} \|_{S} \right |
< \| f_{k+p}-f_{k} \|_{S},
\label{(39)}
\end{equation}
it follows that the sequence of norms
$\{ \| f_{k} \|_{S} \}$ is convergent and hence it is legitimate
to define the norm of $X$ as we have done. Furthermore, from the
definition of the equivalence relation, it follows that the
$S_{1}$-norm of $X$ is independent of the particular representative
chosen to define it. The properties of a norm
$$
\| X \| \geq 0, \;
\| \lambda X \| = |\lambda| \| X \|, \;
\| X+Y \| \leq \| X \| + \| Y \|
$$ 
are proved immediately, including the fact that only the zero vector
has vanishing norm.

Let us now prove the completeness of $S_1$. For this purpose, we
consider a sequence $ \{ X_{n} \}$ of elements of $S_1$ which 
fulfills the Cauchy condition. Given $\varepsilon >0$ one can thus
find an index $\nu_{\varepsilon}$ such that
$$
n> \nu_{\varepsilon} \Longrightarrow \| X_{n+p}-X_{n} \|_{S_{1}}
< {\varepsilon \over 2}, \; \forall p \in N.
$$
If we then denote by $\| f_{k}^{(n)} \|$ a representative of
$X_n$, one can write that
\begin{equation}
n > \nu_{\varepsilon} \Longrightarrow \lim_{k \to \infty}
\| f_{k}^{(n+p)}-f_{k}^{(n)} \|_{S} < {\varepsilon \over 2}, \;
\forall p \in N.
\label{(40)}
\end{equation}
On the other hand, since the sequence $\{ f_{k}^{(n)} \}$ verifies
the Cauchy condition, it is possible to associate to every $n$ an
index $k_n$ such that
\begin{equation}
k > k_n \Longrightarrow \| f_{k}^{(n)}-f_{k_{n}}^{(n)} \|_{S}
< {1 \over n}.
\label{(41)}
\end{equation}
One can also assume, without loss of generality, that $k_{n+1}>k_{n}$.
Let us then consider the sequence 
\begin{equation}
f_{k_{1}}^{(1)}, f_{k_{2}}^{(2)},...,f_{k_{n}}^{(n)},...
\label{(42)}
\end{equation}
and let us show that it verifies the Cauchy condition. For this purpose
we point out that, for all $n$, one finds for all values of $k$
$$
\| f_{k_{n+p}}^{(n+p)}-f_{k_{n}}^{(n)} \|_{S} \leq
\|f_{k_{n+p}}^{(n+p)}-f_{k}^{(n+p)}\|_{S}
+\| f_{k}^{(n+p)}-f_{k}^{(n)} \|_{S}
+ \| f_{k}^{(n)}-f_{k_{n}}^{(n)} \|_{S}.
$$
Thus, if one takes $k > k_{n+p}>k_{n}$, one obtains by virtue of \eqref{(41)}
$$
\| f_{k_{n+p}}^{(n+p)}-f_{k_{n}}^{(n)} \|_{S} \leq
{1 \over (n+p)}+{1 \over n}
+\| f_{k}^{(n+p)}-f_{k}^{(n)} \|_{S}.
$$
Furthermore, upon taking $n > {4 \over \varepsilon}$, one finds
$$
\| f_{k_{n+p}}^{(n+p)} -f_{k_{n}}^{(n)} \|_{S} \leq 
{\varepsilon \over 2}+ \| f_{k}^{(n+p)}-f_{k}^{(n)} \|_{S}.
$$ 
Since $k$ can be chosen as large as we please, it follows
from \eqref{(40)} that 
$$
n > {\rm max} \left({4 \over \varepsilon},\nu_{\varepsilon}\right)
\Longrightarrow 
\| f_{k_{n+p}}^{(n+p)}-f_{k_{n}}^{(n)} \|_{S} \leq 
\varepsilon , \; \forall p \in N.
$$
We have therefore proved that the sequence \eqref{(42)} fulfills the Cauchy
condition and hence it represents an element $X$ of $S_1$. We now
want to prove that
\begin{equation}
\lim_{n \to \infty}X_{n}=X.
\label{(43)}
\end{equation}
Indeed, one has
$$
\| X_{n}-X \|_{S_{1}}=\lim_{r \to \infty} 
\| f_{r}^{(n)}-f_{k_{r}}^{(r)} \|_{S}.
$$
Upon fixing for the moment $n > \nu_{\varepsilon}$, let us
choose $\nu$ anf $s$ in such a way that
$$
r > {\rm max} (n,k_n), \; s > k_{r} \geq k_{n}.
$$
One finds therefore
$$
\| f_{r}^{(n)}-f_{k_{r}}^{(r)} \|_{S} \leq 
\| f_{r}^{(n)}-f_{k_{n}}^{(n)} \|_{S}
+\| f_{k_{n}}^{(n)}-f_{s}^{(n)} \|_{S}
$$
$$
+ \| f_{s}^{(n)}-f_{s}^{(r)} \|_{S}
+ \| f_{s}^{(r)}-f_{k_{r}}^{(r)} \|_{S}.
$$
By virtue of this condition and of the majorization \eqref{(41)}, 
it follows that
$$
\| f_{r}^{(n)}-f_{k_{r}}^{(r)} \|_{S} \leq {2 \over n}
+ \| f_{s}^{(n)}-f_{s}^{(r)} \|_{S}+{1 \over r}.
$$
On passing to the limit as $s \rightarrow \infty$ one obtains,
by virtue of the limit \eqref{(40)},
$$
\| f_{r}^{(n)}-f_{k_{r}}^{(r)} \|_{S}
\leq {2 \over n}+{1 \over r}+{\varepsilon \over 2}.
$$
The subsequent limit as $r \rightarrow \infty$ yields
$$
n > \nu_{\varepsilon} \Longrightarrow 
\| X_{n}-X \|_{S_{1}} \leq {2 \over n}
+{\varepsilon \over 2},
$$
which in turn implies the limit \eqref{(43)}. We have thus proved the
completeness of $S_1$. 

If we now consider the space $\Sigma$ of the $X \in S_1$ for which the
sequence $\{ f_{k} \}$ that represents $X$ is convergent, we
can consider the map
\begin{equation}
x=S \cdot \lim_{k \to \infty} x_{k} \in S \longrightarrow
X= \{ f_{k} \} \in \Sigma,
\label{(44)}
\end{equation}
that we agree to denote by $X=T(x)$. Since by virtue of
the limit \eqref{(38)} one can write that
$$
\| T(x) \|_{S_{1}}= \| x \|_{S},
$$
the equivalence of $S$ and $\Sigma$ has been proved.

Let us now prove that $\Sigma$ is dense upon $S_1$. For this
purpose, let $X$ be an element of $S_1$ while $\{ f_k \}$ is
one of its representatives. Let us set
$$
\forall n \in N, \; h_{n}^{(k)}=f_{k}, \;
Y^{(k)}= \{ h_{n}^{(k)} \}_{n \in N}.
$$
Since the sequences $\{ h_{n}^{(k)} \}_{n \in N}$ are all
constant, one has $Y^{(k)} \in \Sigma$. On the other hand,
by virtue of the Cauchy condition satisfied by $X$, one has
$$
\forall \varepsilon >0 \; \exists \nu_{\varepsilon}:
n,k > \nu_{\varepsilon} \Longrightarrow 
\| f_{n} -h_{n}^{(k)} \|_{S} < \varepsilon
$$
and hence, by passage to the limit as $n \rightarrow \infty$, one finds
$$
k > \nu_{\varepsilon} \Longrightarrow \| X-Y^{(k)} \|_{S_{1}}
> \varepsilon
$$
i.e.
$$
X=\lim_{k \to \infty} Y^{(k)}.
$$

The last remaining task is to prove that, if $S_{1}'$ is another
Banach space containing a space $\Sigma'$ equivalent to $S$ and
dense upon $S_{1}'$, the two spaces $S_1$ and $S_{1}'$ are equivalent. 
Indeed, if one defines an isomorphism between $\Sigma'$ and $\Sigma$ 
by establishing a correspondence between $X' \in \Sigma$ and 
$X \in \Sigma$ when they have the same image in $S$, one has
$$
\| X' \|_{\Sigma'} = \| X \|_{\Sigma}
$$
and hence such an isomorphism is an equivalence between $\Sigma$
and $\Sigma'$. Bearing in mind that $\Sigma$ is dense upon $S_1$ 
and $\Sigma'$ is dense upon $S_{1}'$, such an equivalence is
extended by continuity to an equivalence between $S_1$ and $S_{1}'$.

\subsection{Abstract measure theory}

Following \cite{AFP}, let $X$ be a nonempty set and let $\mathcal E$ be
a collection of subsets of $X$.
\vskip 0.3cm
\noindent
(i) The set $\mathcal E$ is said to be an algebra if the empty set
$\emptyset \in {\mathcal E}$, the union $E_1 \cup E_2 \in {\mathcal E}$
and the set-theoretic difference $X - E_1 \in {\mathcal E}$ 
whenever $E_1$ and $E_2$ belong to $\mathcal E$.
\vskip 0.3cm
\noindent
(ii) An algebra $\mathcal E$ is said to be a $\sigma$-algebra if,
for any sequence $(E_h) \subset {\mathcal E}$, its union
$\cup_{h}E_h$ belongs to $\mathcal E$.
\vskip 0.3cm
\noindent
(iii) For any collection $\mathcal G$ of subsets of $X$, the
$\sigma$-algebra generated by $\mathcal G$ is the smallest
$\sigma$-algebra containing $\mathcal G$. If $(X,\tau)$ is a
topological space, one denotes by ${\mathcal B}(X)$ the 
$\sigma$-algebra of Borel subsets of $X$, i.e., the 
$\sigma$-algebra generated by the open subsets of $X$.
\vskip 0.3cm
\noindent
(iv) If $\mathcal E$ is a $\sigma$-algebra in $X$, the pair
$(X,{\mathcal E})$ is said to be a measure space.
 
By virtue of De Morgan laws, algebras are closed under 
finite intersections, and $\sigma$-algebras are closed
under countable intersections. Furthermore, since the
intersection of any family of $\sigma$-algebras is a
$\sigma$-algebra, the concept of generated 
$\sigma$-algebra is meaningful. Sets endowed with a
$\sigma$-algebra are the appropriate framework to
introduce measures. 

If $(X,{\mathcal E})$ is a measure space, let $m$ be a natural
number $\geq 1$.
\vskip 0.3cm
\noindent
(a) The function $\mu: \; {\mathcal E} \rightarrow 
{\mathbb R}^{m}$ is a measure if $\mu(\emptyset)=0$ and,
for any sequence $(E_h)$ of pairwise disjoint elements
of $\mathcal E$, countable additivity holds, i.e.,
$$
\mu \left(\bigcup_{h=0}^{\infty}E_{h}\right)
=\sum_{h=0}^{\infty}\mu(E_h).
$$
If $m=1$, $\mu$ is said to be a real measure, whereas, if
$m>1$, $\mu$ is said to be a vector measure.
\vskip 0.3cm
\noindent
(b) If $\mu$ is a measure, its total variation $|\mu|$ for
every $E \in {\mathcal E}$ is defined according to
$$
|\mu|(E) \equiv {\rm sup} \left \{ \sum_{h=0}^{\infty}
|\mu(E_h)|: \; E_h \in {\mathcal E} \; {\rm pairwise}
\; {\rm disjoint}, \; E=\bigcup_{h=0}^{\infty}E_h
\right \}.
$$
\vskip 0.3cm
\noindent
(c) If $\mu$ is a real measure, its positive part $\mu^{+}$
and negative part $\mu^{-}$ are defined as follows:
$$
\mu^{+}={{|\mu|+\mu}\over 2}, \;
\mu^{-}={{|\mu|-\mu}\over 2}.
$$

In particular, if the set $X$ is nonempty and $\mathcal E$
is the $\sigma$-algebra of all its subsets, one can define
the Dirac measure on $(X,{\mathcal E})$ as follows. To each 
$x \in X$ we associate the measure $\delta_{x}$ defined
by $\delta_{x}(E)=1$ if $x \in E$, $\delta_{x}(E)=0$ otherwise.
If $(x_h)$ is a sequence in $X$ and if $(c_h)$ is a sequence in
${\mathbb R}^{m}$ such that the series $\sum_{h}|c_h|$ 
is convergent, we can set
$$
\left(\sum_{h=0}^{\infty}c_{h}\delta_{x_{h}} \right)(E)
=\sum_{x_{h} \in E} c_{h}
$$
and obtain a ${\mathbb R}^{m}$-valued measure. Measures of this
kind are said to be {\it purely atomic}. The set $S_{\mu}$
of atoms of a measure $\mu$ in a measure space 
$(X,{\mathcal E})$ is defined by
$$
S_{\mu}= \{ x \in X: \; \mu(\{ x \}) \not =0 \},
$$
provided that the singletons $\{ x \}$ are elements of
${\mathcal E}$. If $\mu$ is finite or $\sigma$-finite, the set
of atoms is at most countable.

\subsection{Functions of bounded variation}

If $\Omega$ is a generic open set in ${\mathbb R}^{N}$,
and if $u \in L^{1}(\Omega)$, we say that $u$ is a function
of bounded variation in $\Omega$ if the distributional 
derivative of $u$ is representable by a measure in $\Omega$,
i.e., if \cite{AFP}
\begin{equation}
\int_{\Omega}
u {\partial \phi \over \partial x_{i}}dx=- \int_{\Omega}
\phi dD_{i}u, \; \forall \phi \in C_{c}^{\infty}(\Omega), \;
i=1,...,N
\label{(45)}
\end{equation}
for some ${\mathbb R}^{N}-valued$ measure
$Du=(D_{1}u,...,D_{N}u)$ in $\Omega$. The vector space of
all functions of bounded variation in $\Omega$ is denoted 
by $BV(\Omega)$.

A smoothing argument shows that the integration by parts
just written is still true for any $\phi \in 
C_{c}^{1}(\Omega)$, or even for Lipschitz functions $\phi$
with compact support in $\Omega$. All these formulae can
be written concisely in the form
\begin{equation}
\int_{\Omega}u {\rm div}\varphi \; dx =-\sum_{i=1}^{N}
\int_{\Omega} \varphi_{i} dD_{i}u \; \forall \varphi
\in [C_{c}^{1}(\Omega]^{N}
\label{(46)}
\end{equation}
The same notation can be used also for functions 
$u \in [BV(\Omega)]^{m}$. In such a case, $Du$ is a
$m \times N$ matrix of measures $D_{i}u^{\alpha}$ in
$\Omega$ satisfying
\begin{equation}
\sum_{\alpha=1}^{m}\int_{\Omega} u^{\alpha}
{\rm div} \varphi^{\alpha} dx = -\sum_{\alpha=1}^{m}
\sum_{i=1}^{n} \int_{\Omega} \varphi_{i}^{\alpha}
dD_{i} u^{\alpha}, \; \forall \varphi \in
[C_{c}^{1}(\Omega)]^{mN}.
\label{(47)}
\end{equation}

The Sobolev space $W^{1,1}(\Omega)$ is contained in
$BV(\Omega)$, and the inclusion is strict. If $u$ belongs 
to $[BV_{\rm loc}(\Omega)]^{m}$, one can prove the
following properties:
\vskip 0.3cm
\noindent
(a) If the distributional derivative $Du$ vanishes, $u$
is equivalent to a constant in any connected component 
of $\Omega$.
\vskip 0.3cm
\noindent
(b) For any bounded Lipschitz function 
$\psi: \Omega \rightarrow {\mathbb R}$, the function
$u \psi$ belongs to $[BV_{\rm loc}(\Omega)]^{m}$ and
$$
D(u \psi)=\psi Du + (u \otimes \nabla \psi)
{\mathcal L}^{N},
$$
where $\nabla$ is the approximate pointwise differential
and ${\mathcal L}^{N}$ denotes $N$-dimensional Lebesgue measure.
\vskip 0.3cm
\noindent
(c) If $\rho$ is any convolution kernel and
$$
\Omega_{\varepsilon}= \{ x \in \Omega: \;
{\rm dist}(x,\partial \Omega)> \varepsilon \},
$$
then
$$
\nabla(u \star \rho_{\varepsilon})
=Du \star \rho_{\varepsilon} \;
{\rm in} \; \Omega.
$$

\section*{Declaration} 
The authors certify that:
\begin{enumerate}
\item they have no relevant financial or non-financial interests to disclose;
\item they have no conflicts of interest to declare that are relevant to the content of this manuscript;
\item they  have no financial or proprietary interests in any material discussed in this manuscript; and
\item they have no affiliations with or involvement in any organization or entity 
with any financial interest or non-financial interest in the subject matter or 
materials discussed in this manuscript.
\end{enumerate}

\end{document}